\documentclass[twocolumn,numbook,nospthms,natbib]{svjour3}          
\smartqed  
 \usepackage{graphicx}
\usepackage{amsmath,amssymb,amsthm} 
\usepackage{enumerate,epstopdf}
\usepackage{color}
\usepackage{pdfsync}
\usepackage{float}
\usepackage{graphicx}
\usepackage{xspace}
\usepackage{algpseudocode}
\usepackage{algorithmicx}
\usepackage{algorithm}

\newcommand{\argmin}{\operatornamewithlimits{argmin}}
\newcommand{\argmax}{\operatornamewithlimits{argmax}}

\def\cB{{\mathcal B}}

\def\cF{{\mathcal F}}

\def\cP{{\mathcal P}}

\def\bR{{\mathbb R}}

\def\bE{{\mathbb E}}

\def\bN{{\mathbb N}}

\def\NPDF{{\mathcal N}}

\def\cI{{\mathcal I}}

\def\f0{{\mathbf 0}}

\def\sk{{\mathsf k}}
\def\sp{{\mathsf p}}
\def\sD{{\mathsf D}}
\def\sC{{\mathsf C}}

\def\md{{\mathrm{d}}}

\renewcommand\qedsymbol{$\blacksquare$}

\newtheorem{thm}{Theorem}

\newtheorem{corollary}{Corollary}
\newtheorem{prop}{Proposition}

\newtheorem{assmp}{Assumption}

\theoremstyle{definition}

\newtheorem{rem}{Remark}

\begin{document}

\title{Parallel sequential Monte Carlo for stochastic gradient-free nonconvex optimization\thanks{An important part of this work was carried out when \"O.~D.~A. was visiting Department of Mathematics, Imperial College London. This work was partially supported by \textit{Agencia Estatal de Investigaci\'on} of Spain (RTI2018-099655-B-I00 CLARA), and the regional government of Madrid (program CASICAM-CM S2013/ICE-2845). The work of the second author has been partially supported by a UC3M-Santander Chair of Excellence grant held at the Universidad Carlos III de Madrid.}
}


\author{\"Omer Deniz Akyildiz         \and Dan Crisan \and
        Joaqu\'in M\'iguez 
}


\institute{\"O.~D.~Akyildiz \at
              {University of Warwick}\\
              The Alan Turing Institute\\
              \email{omer.akyildiz@warwick.ac.uk}           
		 \and
		 D.~Crisan \at
              {Imperial College London}\\
              \email{d.crisan@imperial.ac.uk}           
           \and
           J.~M\'iguez \at
              Universidad Carlos III de Madrid \& Instituto de Investigaci\'on Sanitaria Gregorio Mara\~n\'on\\
              \email{joaquin.miguez@uc3m.es} 
}

\date{Received: date / Accepted: date}

\maketitle

\begin{abstract}
We introduce and analyze a parallel sequential Monte Carlo methodology for the numerical solution of optimization problems that involve the minimization of a cost function that consists of the sum of many individual components. The proposed scheme is a stochastic zeroth order optimization algorithm which demands only the capability to evaluate small subsets of components of the cost function. It can be depicted as a bank of samplers that generate particle approximations of several sequences of probability measures. These measures are constructed in such a way that they have associated probability density functions whose global maxima coincide with the global minima of the original cost function. The algorithm selects the best performing sampler and uses it to approximate a global minimum of the cost function. We prove analytically that the resulting estimator converges to a global minimum of the cost function almost surely and provide explicit convergence rates in terms of the number of generated Monte Carlo samples and the dimension of the search space. We show, by way of numerical examples, that the algorithm can tackle cost functions with multiple minima or with broad ``flat'' regions which are hard to minimize using gradient-based techniques.
\end{abstract}

\tolerance = 1600
\pretolerance = 1600


\section{Introduction}

In signal processing and machine learning, optimization problems of the form
\begin{align}\label{eqCostFunction}
\min_{\theta\in\Theta} f(\theta) = \sum_{i=1}^n f_i(\theta),
\end{align}
where $\Theta \subset \bR^d$ is the $d$-dimensional compact search space, have attracted significant attention in recent years for problems where $n$ is very large. Such problems often arise in big data settings, e.g., when one needs to estimate parameters given a {large number of observations \citep{bottou2018optimization}}. 

Because of their efficiency, the optimization community has focused mainly on stochastic gradient based methods \citep{RobbinsMonro,adaGrad,kingma2014adam} (see \citet{bottou2018optimization} for a recent review of the field) where an estimate of the gradient is obtained using a randomly selected subsample of the gradients of the component functions (the $f_i$'s in Eq. \eqref{eqCostFunction}) at each iteration. The resulting estimate is then used to perform a stochastic descent step. The majority of these stochastic gradient methods construct the subsamples using sampling with replacement to obtain unbiased estimates of the gradient. The latter can then be seen as a noisy gradient estimate with additive, zero-mean noise. In practice, however, there are schemes that subsample the data set without replacement (hence producing biased gradient estimators) and it has been argued that such methods can attain better numerical performance \citep{gurbuzbalaban2015random, shamir2016without}.

The gradient information may not be always available, however, due to different reasons. For example, in an engineering application, the system to be optimized might be a black-box, e.g., a piece of closed software code with free parameters, which can be \textit{evaluated} but cannot be differentiated \citep{nesterov2011random}. In these cases, one needs to use a \textit{gradient-free} optimization scheme, meaning that the scheme must rely only on function evaluations, rather than any sort of actual gradient information. \textit{Classical} gradient-free optimization methods have attracted significant interest over the past decades \citep{appel2004accelerated,spall2005introduction,marino2007monte,conn2009introduction}. These methods proceed either by a random search (which is based on evaluating the cost function at random points and update the parameter whenever a descent in the function evaluation is achieved \citep{spall2005introduction}), or by constructing a numerical (finite-difference type) approximation of the gradient that can be used to take a descent step \citep{nesterov2011random}.

Such methods are not applicable, however, if one can only obtain noisy function evaluations or one can only evaluate certain subsets of component functions in a problem like \eqref{eqCostFunction}. In this case, since the function evaluations are not exact, direct random search methods cannot be used reliably and it is only recently that some authors have described how to compute finite-difference approximations of the gradient \citep{wibisono2012finite, ghadimi2013stochastic, chen2015randomized, bach2016highly}. Also in recent years, evolutionary methods, based on the mutation, recombination and selection of samples, have been suggested for the approximation of gradients. The resulting optimization algorithms, termed evolutionary strategies (ES) have been applied within reinforcement learning schemes \citep{salimans2017evolution,wierstra2014natural,hansen2001completely,morse2016simple}.

However, {when the cost function has multiple minima or has some regions where the gradient vanishes, gradient-based methods may suffer from poor numerical performance.} In particular, the optimizer can get stuck in a local minimum easily, due to its reliance on gradient approximations. Moreover, when the gradient contains little information about any minimum (e.g., in flat regions), gradient-free stochastic optimizers (as well as perfect gradient schemes) can suffer from slow convergence.

{Model-based random-search methods \citep{hu2012survey}, which use probabilistic models of various types in order to speed up the search procedure, have been investigated in order to address problems where gradients cannot be approximated or simply turn out ineffective. The latter include classical algorithms such as simulated annealing (SA) \citep{kirkpatrick1983optimization}, Monte Carlo expectation maximization (EM) \citep{robert2004monte} and other Markov chain Monte Carlo (MCMC) based methods \citep{pereyra2015survey}. The class of model-based random search schemes also encompasses sequential Monte Carlo (SMC) techniques, e.g., SMC implementations of SA \citep{zhou2013sequential} and several optimization algorithms that mimic standard particle filters \citep{zhou2013particle,liu2016particle}. Let us note that most of the latter MCMC- and SMC-based procedures can be cast within the class of SMC samplers described in \citet{del2006sequential}, albeit with a target distribution which is sometimes implicitly defined in order to satisfy certain properties related to the objective function \citep{zhou2013particle}. Nevertheless, these optimization techniques are generally designed to be used in problems where the objective function can be evaluated exactly and their extension to stochastic optimization is not straightforward, neither from the point of view of practical performance nor in terms of theoretical convergence guarantees.}

{Some authors have also explored the duality between optimization and probability theory, in a way that potentially enables the use of general computational inference algorithms for solving optimization problems. While in model-based optimization the emphasis is put on the algorithms (e.g., how to use MCMC methods in \citet{pereyra2015survey} or particle filters in \citet{liu2016particle}, for optimization), in this line of research the emphasis is in converting the optimization problem into an equivalent inference problem, which can then be tackled with {\em any} suitable inference algorithm. A rigorous mathematical treatment of the topic can be found in \citet{del1999maslov}, while \citet{ikonen2005application} and \citet{miguez2013convergence} address the problem from a methodological viewpoint. Again, these contributions deal with problems where the objective function can be computed deterministically and exactly, though.}

{The stochastic setting, where it is only possible to compute noisy evaluations of $f(\theta)$, is harder and the bibliography is limited in comparison with the deterministic setup. The recent survey in \citet{homem2014monte} covers various gradient-based Monte Carlo procedures, however it addresses a different class of stochastic optimizaton problems where the cost function itself is defined as an expectation, rather than a finite-sum as in \eqref{eqCostFunction}.  Existing model-based random methods for stochastic optimization include MCMC-based samplers which target a probability density function (pdf) matched to the  objective function in \eqref{eqCostFunction} (meaning that the maxima of the pdf coincide with the minima of $f(\theta)$) \citep{welling2011bayesian,chen2016bridging}. Such schemes, however, also rely on the computation of noisy gradients. Other MCMC-based methods (see, e.g., \citet{alquier2016noisy} which employs noisy Metropolis steps) do not require gradients, yet these techniques have been primarily designed and investigated as sampling algorithms, rather than optimization methods. Similarly, an adaptive importance sampler for a target pdf matched to $f(\theta)$ is reported in \citet{akyildiz2017adaptive}. This method uses subsampling to compute noisy weights, but the technique lacks any theoretical guarantees and does not address the problem optimization directly either. A particle filtering algorithm for stochastic global optimization has been proposed by \citet{stinis2012stochastic}. The method is intuitive, simple to implement and has been shown to work efficiently in some simple examples, however the contribution of  \citet{stinis2012stochastic} is strictly methodological: there is no analysis of performance and no theoretical guarantees.}

In this paper, we propose a \textit{parallel sequential Monte Carlo optimizer} (PSMCO) to minimize cost functions with the finite-sum structure of problem \eqref{eqCostFunction}. The PSMCO is a zeroth-order stochastic optimization algorithm, in the sense that it only uses evaluations of small batches of individual components $f_i(\theta)$ in \eqref{eqCostFunction}. {In particular, it does not require the computation or approximation of gradients.} The proposed scheme proceeds by constructing parallel samplers, each of which aims at minimizing the same cost function $f(\theta)$. Each sampler performs subsampling without replacement to obtain its mini-batches of individual components and processes each component only once. Using these mini-batches, the PSMCO constructs potential functions, propagates samples via a jittering scheme \citep{crisan2018nested} and selects samples by applying a weighting-resampling procedure. The communication between parallel samplers is only necessary when a joint estimate of the minimum is required. In this case, the best performing sampler is selected and the minimum is estimated. 

{We analytically prove that the estimate provided by each sampler converges almost surely to a global minimum of the cost function and provide explicit convergence rates in terms of the number of Monte Carlo samples generated by the algorithm. This type of analysis goes beyond standard results for particle filters: it tackles the problem of stochastic optimization directly and it yields stronger theoretical guarantees compared to other stochastic optimization methods in the literature. In particular, we obtain error bounds for the solution of problem \eqref{eqCostFunction} that hold almost surely (a.s.) and vanish at a rate $\mathcal{O}\left( N^{-\frac{1}{2(d+1)}} \right)$,
where $N$ is the number of Monte Carlo samples and $d$ is the dimension of the search space $\Theta$. This is in contrast to the usual results for random search methods in the literature, which are purely asymptotic and do not provide any rates \citep{appel2004accelerated,miguez2010analysis,hu2012survey,zhou2013sequential,zhou2013particle}. Let us also remark the difference between the proposed scheme and the SMC-based schemes in \citet{miguez2013convergence} where the authors partitioned the parameter vector and modeled it as a dynamical system, an approach that cannot be used in the more general setup of \eqref{eqCostFunction} because each individual function $f_i$ depends on the complete vector $\theta$. The PSMCO algorithm, in turn, is explicitly designed to provide an estimate of the full parameter $\theta$ at each iteration.}

{The main contribution of this paper includes the theoretical analysis of the proposed PSMCO scheme and its numerical demonstration on three problems where classical stochastic optimization methods (especially gradient-based algorithms) struggle to perform. The paper is organized as follows. After a brief survey of the relevant notation (below), we lay out the relationship between Bayesian inference and optimization in Section \ref{sec:OptSamp}. Then, we develop a sequential Monte Carlo scheme in Section \ref{sec:TheAlg}. In Section \ref{sec:Analysis}, we analyze this scheme and investigate its theoretical properties. We present some numerical results in Section \ref{sec:Exp} and make some concluding remarks in Section \ref{sec:Conc}.}

\subsection*{Notation}

For $n\in\bN$, we denote $[n] = \{1,\ldots,n\}$. The space of bounded functions on the parameter space $\Theta \subset \bR^d$ is denoted as $B(\Theta)$. The set of continuous and bounded real functions on $\Theta$ is denoted $\mathsf{C}_b(\Theta)$. The family of Borel subsets of $\Theta$ is denoted with $\cB(\Theta)$. The set of probability measures on the measurable space $(\Theta,\cB(\Theta))$ is denoted $\cP(\Theta)$. Given $\varphi\in B(\Theta)$ and $\pi\in\cP(\Theta)$, the integral of $\varphi$ with respect to (w.r.t.) $\pi$ is written as $$(\varphi,\pi) = \int_\Theta \varphi(\theta) \pi(\mbox{d}\theta).$$ Given a Markov kernel $\kappa:\cB(\Theta) \times \Theta \mapsto [0,1]$, we denote $\kappa \pi(\md\theta) = \int \kappa(\mbox{d}\theta | \theta') \pi(\mbox{d}\theta')$. If $\varphi\in B(\Theta)$, then $\|\varphi\|_\infty = \sup_{\theta\in\Theta} |\varphi(\theta)| < \infty$.

Let $\alpha = (\alpha_1,\ldots,\alpha_{d}) \in \bN^* \times \cdots \times \bN^*$, where $\bN^* = \bN \cup \{0\}$, be a multi-index. We define the partial derivative operator $\sD^\alpha$ as
\begin{align*}
\sD^\alpha h = \frac{\partial^{\alpha_1} \cdots \partial^{\alpha_d} h}{\partial \theta_1^{\alpha_1} \cdots \partial \theta_d^{\alpha_d}}
\end{align*}
for a sufficiently differentiable function $h:\bR^d \to \bR$. We use $|\alpha| = \sum_{i=1}^d \alpha_i$ to denote the order of the derivative. Finally, the notation $\lfloor x \rfloor$ indicates the floor function for a real number $x$, which returns the biggest integer $k \leq x$.

\section{Stochastic optimization as inference}\label{sec:OptSamp}

In this section, we describe how to construct a sequence of probability distributions that can be linked to the solution of problem~\eqref{eqCostFunction}. Let $\pi_0\in\cP(\Theta)$ be the initial element of the sequence. We construct the rest of the sequence recursively as
\begin{align}\label{eqDistFlow}
\pi_t(\mbox{d}\theta) = \pi_{t-1}(\md \theta) \frac{G_t(\theta)}{\int_\Theta G_t(\theta) \pi_{t-1}(\md \theta)}, \quad \textnormal{for } t\geq 1,
\end{align}
where the maps $G_t:\Theta \mapsto \bR_+$ are termed \textit{potential functions} \citep{del2004feynman}. The key idea is to associate these potentials $(G_t)_{t\geq 1}$ with mini-batches of individual components of the cost function (subsets of the $f_i$'s) in order to construct a sequence of measures $\pi_0,\pi_1,\ldots,\pi_T$ such that (for a prescribed value of $T$) the global maxima of the density of $\pi_T$ match the global minima of $f(\theta)$. We remark that the measures $\pi_1,\ldots,\pi_T$ are all absolutely continuous w.r.t $\pi_0$ if the potential functions $G_t$, $t = 1,\ldots,T$, are bounded.


To construct the potentials, we use mini-batches consisting of $K$ individual functions $f_i$ for each iteration $t$. To be specific, we randomly select subsets of indices $\cI_t, t = 1,\ldots,T$, by drawing uniformly from $\{1,\ldots,n\}$ without replacement. Each subset has $|\cI_t| = K$ elements, in such a way that we obtain $T$ subsets satisfying $\bigcup_{i=1}^T \cI_t = [n]$ and $\cI_i \cap \cI_j = \emptyset$ when $i\neq j$. Finally, we define the potential functions $(G_t)_{t\geq 1}$ as
\begin{align}\label{eq:potential}
G_t(\theta) = \exp\left(-\sum_{i \in \cI_t} f_i(\theta)\right), \quad \quad t = 1,\ldots,T.
\end{align}
In the sequel, we provide a result that establishes a precise connection between the optimization problem in \eqref{eqCostFunction} and the sequence of probability measures defined in \eqref{eqDistFlow}, provided that Assumption~\ref{assmp:integrability} below is satisfied.

\begin{assmp}\label{assmp:integrability} The functions in the sequence $(G_t)_{t\geq 1}$ are positive and bounded, i.e.,
\begin{align*}
G_t(\theta) > 0 \quad \forall \theta\in\Theta \quad \textnormal{and} \quad G_t\in B(\Theta).
\end{align*}
\end{assmp}

Next, we show the relationship between the minima of $f(\theta)$ and the maxima of $\frac{\md \pi_T}{\md \pi_0}$.
\begin{prop}\label{prop:RN} Assume that the potentials are selected as in \eqref{eq:potential} for $1\leq t \leq T$, with $\cI_i \cap \cI_j = \emptyset$ and $\bigcup_i \cI_i = [n]$. Let $\pi_T$ be the $T$-th probability measure constructed by means of recursion \eqref{eqDistFlow}. If Assumption~\ref{assmp:integrability} holds and $\pi_0\in\cP(\Theta)$, then
\begin{align*}
\argmax_{\theta\in\Theta} \frac{\md\pi_T}{\md\pi_0}(\theta) = \argmin_{\theta\in\Theta} \sum_{i=1}^n f_i(\theta),
\end{align*}
where $\frac{\md\pi_T}{\md\pi_0}(\theta):\Theta \to \bR_+$ denotes the Radon-Nikodym derivative of $\pi_T$ w.r.t. the prior measure $\pi_0$.
\end{prop}
\begin{proof}
See Appendix~\ref{proof:RN}.
\end{proof}

For conciseness, we abuse the notation and use $\pi(\theta)$, $\theta\in\Theta$, to indicate the pdf associated to a probability measure $\pi(\md\theta)$. The two objects are distinguished clearly by the context (e.g., for an integral $(\varphi,\pi)$, $\pi$ necessarily is a measure) but also by their arguments. The probability measure $\pi(\cdot)$ takes arguments $\md\theta$ or $A \in \cB(\Theta)$, while the pdf $\pi(\theta)$ is a function $\Theta\to [0,\infty)$.
\begin{rem}\label{rem:Prior} 
Notice that when $\pi_0$ is a uniform probability measure on $\Theta$, we simply have
\begin{align*}
 \pi_T(\theta) \propto \exp\left(-\sum_{i=1}^n f_i(\theta)\right), \quad \theta\in\Theta.
\end{align*}
where $\pi_T(\theta)$ denotes the pdf (w.r.t. Lebesgue measure) of the measure $\pi_T(\md\theta)$. \hfill $\blacksquare$
\end{rem}
{\begin{rem} 
Moreover, if we choose
\begin{align}\label{eq:PriorFromPotentials}
\pi_0(\theta) \propto \exp\left(-f_1(\theta)\right)
\end{align}
and select index subsets such that $\bigcup_{t=1}^T \cI_t = \{2,\ldots,n\}$ then we also obtain
\begin{align*}
 \pi_{T}(\theta) \propto \exp\left(-\sum_{i=1}^n f_i(\theta)\right), \quad \quad \textnormal{for } \theta\in\Theta.
\end{align*}
When a Monte Carlo is scheme used to realize recursion \eqref{eqDistFlow}, the use of a prior of the form \eqref{eq:PriorFromPotentials} requires the ability to sample from it.
\hfill $\blacksquare$	
\end{rem}}

In summary, if we can construct the sequence described by \eqref{eqDistFlow}, then we can replace the minimization problem of $f(\theta)$ in \eqref{eqCostFunction} by the maximization of a pdf. This relationship was exploited in a Gaussian setting in \citet{akyildiz2018proximal}, i.e., the special case of a Gaussian prior $\pi_0$ and log-quadratic potentials $(G_t)_{t\geq 1}$ (Gaussian likelihoods), which makes it possible to implement recursion \eqref{eqDistFlow} analytically. The solution of this special case can be shown to match a well-known stochastic optimization algorithm, called the incremental proximal method \citep{bertsekas2011incremental}, with a variable-metric. However, for general priors and potentials, it is not possible to analytically construct \eqref{eqDistFlow} and maximize $\pi_T(\theta)$. For this reason, we propose a simulation method to approximate the recursion \eqref{eqDistFlow} and solve $\argmax_{\theta\in\Theta} \frac{\md \pi_T}{\md \pi_0}(\theta)$.

\section{The algorithm}\label{sec:TheAlg}

In this section we first describe a sampler to simulate from the distributions defined by recursion \eqref{eqDistFlow}. We then describe an algorithm which runs these samplers in parallel. The parallelization here is not primarily motivated by the computational gain (although it can be substantial). We have empirically found that non-interacting parallel samplers are able to keep track of multiple minima better than a single ``big'' sampler. For this reason, we will not focus on demonstrating computational gains in the experimental section. Rather, we will discuss what parallelization brings in terms of providing better estimates.

We consider $M$ workers (corresponding to $M$ samplers). Specifically, each worker sees a different configuration of the dataset, i.e., the $m$-th worker constructs a distinct sequence of index sets $(\cI_t^{(m)})_{t\geq 1}$ which determine the mini-batches sampled from the full set of individual components. Having obtained different mini-batches which are randomly constructed, each worker then constructs different potentials $(G_t^{(m)})_{t\geq 1}$, where $G_t^{(m)}(\cdot) = \exp\left\{ -\sum_{i\in\cI_t} f_i(\cdot) \right\}$, as described in the previous section.

The $m$-th worker, therefore, aims at estimating a specific sequence of probability measures $\pi_t^{(m)}$, for $m \in \{1,\ldots,M\}$. We denote the particle approximation of the posterior $\pi_t^{(m)}$ at time $t$ as 
\begin{align*}
\pi_t^{(m),N}(\md\theta) = \frac{1}{N} \sum_{i=1}^N \delta_{\theta^{(i,m)}}(\md\theta),
\end{align*}
where $\delta_{\theta'}(\md \theta)$ is the unit delta measure located at $\theta' \in \Theta$. Overall, the algorithm retains $M$ probability distributions. Note that these distributions are different for each $t<T$, as they depend on different potentials, but $\pi_T^{(m)}=\pi_T$ for all workers because $\bigcup_{t=1}^T \cI_t^{(m)} = [n]$ for every $m$.

\begin{algorithm}[t] 
\begin{algorithmic}[1] 
\caption{Sampler on a local node $m$}\label{alg:localnode}
\State Sample $\theta_0^{(i,m)} \sim \pi_0$ for $i = 1, \ldots, N$.
\For{$t\geq 1$}
\State Jitter by generating samples
\begin{align*}
\hat{\theta}_t^{(i,m)} \sim \kappa(\mbox{d}\theta | \theta_{t-1}^{(i,m)}) \quad\quad \quad \textnormal{\,\,\,\,\,\,for} \quad i = 1,\ldots,N.
\end{align*}
\State Compute normalized weights,
\begin{align*}
w_t^{(i,m)} = \frac{G_t^{(m)}(\hat{\theta}_t^{(i,m)})}{\sum_{i=1}^N G_t^{(m)}(\hat{\theta}_t^{(i,m)})}  \quad \textnormal{for} \quad i = 1,\ldots,N.
\end{align*}
\State Resample by drawing $N$ i.i.d. samples,
\begin{align*}
\theta_t^{(i,m)} \sim \hat{\pi}_t^{(m),N}(\md\theta) := \sum_{i=1}^N w_t^{(i,m)} \delta_{\hat{\theta}_t^{(i,m)}}(\mbox{d}\theta),
\end{align*}
for $ i = 1,\ldots,N.$
\EndFor
\end{algorithmic}
\end{algorithm}

One iteration of the algorithm on a local worker $m$ can be described as follows. Assume the worker has computed the probability measure $\pi_{t-1}^{(m),N}$ using the particle system $\{\theta_{t-1}^{(m,i)}\}_{i=1}^N$. First, we use a jittering kernel $\kappa(\md \theta | \theta_{t-1})$ (a Markov kernel on $\Theta$) to modify the particles \citep{crisan2018nested} (see Section~\ref{sec:JitteringKernel} for the precise definition of $\kappa(\cdot|\cdot)$). The idea is to \textit{jitter} a subset of the particles in order to modify and propagate them into better regions of $\Theta$ with higher probability density and lower cost. The particles are jittered by sampling,
\begin{align*}
\hat{\theta}_t^{(i,m)} \sim \kappa(\cdot | \theta_{t-1}^{(i,m)}) \quad \textnormal{for } i = 1,\ldots,N.
\end{align*}
Note that the jittering kernel may be designed so that it only modifies a subset of particles (again, see Section~\ref{sec:JitteringKernel} for details). Next, we compute weights for the new set of particles $\{\hat{\theta}_t^{(i,m)}\}_{i=1}^N$ according to the $t$-th potential, namely
\begin{align*}
w_t^{(i,m)} = \frac{G_t^{(m)}(\hat{\theta}_t^{(i,m)})}{\sum_{i=1}^N G_t^{(m)}(\hat{\theta}_t^{(i,m)})}  \quad \textnormal{for} \quad i = 1,\ldots,N.
\end{align*}

\begin{rem}
The particle weights can be made proportional to the potentials alone, i.e., $w_t^{(i,m)} \propto G_t^{(m)}(\hat \theta_t^{(i,m)})$, as long as the jittering kernels satisfy Assumption \ref{assmp:kernel} in Section \ref{sec:JitteringKernel}. Under mild assumptions, Algorithm 1 converges with standard error rates $\mathcal{O}(N^{-\frac{1}{2}})$, as proved in Section \ref{sec:Analysis}.
\end{rem}

After obtaining weights, each worker performs a resampling step where for $i = 1,\ldots,N$, we set $\theta_t^{(i,m)} = \hat{\theta}_t^{(i,k)}$ for $k \in \{1,\ldots,N\}$ with probability $w_t^{(i,m)}$. The procedure just described corresponds to a simple multinomial resampling scheme, but other standard methods can be applied as well \citep{douc2005comparison}.
We denote the resulting probability measure constructed at the $t$-th iteration of the $m$-th worker as
\begin{align*}
\pi_t^{(m),N}(\mbox{d}\theta) = \frac{1}{N} \sum_{i=1}^N \delta_{{\theta}_t^{(i,m)}}(\mbox{d}\theta).
\end{align*}

The full procedure for the $m$-th worker is outlined in Algorithm~\ref{alg:localnode}. In Section~\ref{sec:JitteringKernel}, we elaborate on the selection of the jittering kernels and in Section~\ref{sec:estMin}, we detail the scheme for estimating a global minimum of $f(\theta)$ from the set of random measures $\{\pi_t^{(m),N}\}_{m=1}^M$.
%

\subsection{Jittering kernel}\label{sec:JitteringKernel}

The jittering kernel constitutes one of the key design choices of the proposed algorithm. Following \citet{crisan2018nested}, we put the following assumption on the kernel $\kappa$.
\begin{assmp}\label{assmp:kernel} The Markov kernel $\kappa$ satisfies
\begin{align*}
\sup_{\theta' \in \Theta} \int_\Theta |\varphi(\theta) - \varphi(\theta')| \kappa(\mbox{d}\theta|\theta') \leq \frac{c_\kappa \|\varphi\|_\infty}{\sqrt{N}}
\end{align*}
for any $\varphi \in B(\Theta)$ and some constant $c_\kappa < \infty$ independent of $N$.
\end{assmp}
In this paper, we use kernels of form
\begin{align}\label{eq:kernelDefn}
\kappa(\mbox{d}\theta | \theta') = (1 - \epsilon_N) \delta_{\theta'}(\mbox{d}\theta) + \epsilon_N \tau(\mbox{d}\theta | \theta'),
\end{align}
where $\epsilon_N \leq \frac{1}{\sqrt{N}}$, which satisfy Assumption~\ref{assmp:kernel} \citep{crisan2018nested}. The kernel $\tau$ can be rather simple, such as a multivariate Gaussian or multivariate-t distribution centered around $\theta' \in \Theta$. Other choices of $\tau$ are possible as well.
\begin{rem} {The design of the kernel as a centered Gaussian or a multivariate-t distribution around $\theta'$ may not guarantee the propagation of samples into better (lower cost) regions. In this case, the weighting-and-resampling procedure naturally tends to keep and replicate the particles that attain a lower cost. However, the jittering kernel can also be designed to accelerate the optimization process. In particular, our setup allows for the use of gradient estimators (such as finite-difference schemes \citep{nesterov2011random} or nudging steps \citep{akyildiz2020nudging}) in the jittering kernel to accelerate the propagation of samples into lower-cost regions.}
\end{rem}
\subsection{Estimating the global minima of $f(\theta)$}\label{sec:estMin}
In order to estimate the global minima of $f(\theta)$, we first assess the performance of the samplers run by each worker. A typical performance measure is the \textit{marginal likelihood estimate} resulting from $\pi_t^{(m),N}$. After choosing the worker which has attained the highest marginal likelihood (say the $m_0$-th worker), we estimate a minimum of $f(\theta)$ by selecting the particle $\theta_t^{(i,m)}$ that yields the highest density $\pi_t^{(m_0)}(\theta_t^{(i,m_0)})$.
\begin{algorithm}[thb]
\begin{algorithmic}[1]
\caption{PSMCO}\label{alg:global}
\State Sample $\theta_0^{(i,m)} \sim \pi_0$ for $i = 1, \ldots, N$.
\For{$t\geq 1$}
\For{$m = 1,\ldots,M$}
\State {Jitter $N$ local particles (step~3 of Algorithm~\ref{alg:localnode}).}
\State Update the marginal likelihood
\begin{align*}
Z_{1:t}^{(m),N} = Z_{1:t-1}^{(m),N} \times Z_t^{(m),N}
\end{align*}
where
\begin{align*}
Z_t^{(m),N} = \frac{1}{N} \sum_{i=1}^N G_t^{(m)}(\hat{\theta}_t^{(i,m)}).
\end{align*}
\State {Compute weights (step~4 of Algorithm~\ref{alg:localnode}).}
\State {Resample (step~5 of Algorithm~\ref{alg:localnode}).}
\EndFor
\If{an estimate of the solution of problem \eqref{eqCostFunction} is needed at time $t$}
\State Choose
$$m_t^\star = \argmax_{m \in \{1,\ldots,M\}} Z_{1:t}^{(m),N}$$
\State Estimate
\begin{align*}
\theta^{\star,N}_t = \argmax_{i\in\{1,\ldots,N\}} \mathsf{p}^{(m_t^\star),N}_t(\theta_t^{(i,m_t^\star)}).
\end{align*}
\EndIf
\EndFor
\end{algorithmic}
\end{algorithm}

To be precise, let us start by denoting the \textit{incremental} marginal likelihood associated to $\pi_t^{(m)}$ and its estimate $\pi_t^{(m),N}$ as $Z_{1:t}^{(m)}$ and $Z_{1:t}^{(m),N}$, respectively. They can be explicitly obtained by first computing
\begin{align*}
Z_{t}^{(m)} &= \int G_t^{(m)}(\theta) \hat{\pi}^{(m)}_t(\mbox{d}\theta)\\
& \approx \frac{1}{N} \sum_{i = 1}^N G_t^{(m)}(\hat{\theta}_t^{(i,m)}) := Z_t^{(m),N}
\end{align*}
and then updating the running products
\begin{align*}
Z_{1:t}^{(m)} = Z_t^{(m)} Z_{1:t-1}^{(m)} = \prod_{k=1}^t Z_k^{(m)}
\end{align*}
and
\begin{align*}
Z_{1:t}^{(m),N} = Z_t^{(m),N} Z_{1:t-1}^{(m),N} = \prod_{k=1}^t Z_k^{(m),N}.
\end{align*}
The quantity $Z_{1:t}^{(m)}$ is a local performance index that keeps track of the ``quality'' of the $m$-th particle system $\{\theta_t^{(i,m)}\}_{i=1}^N$ \citep{elvira2017adapting} and, hence, we use $\{Z_{1:t}^{(m),N}\}_{m=1}^M$ to determine the best performing worker\footnote{{If we interpret each sequence of index sets $(\mathcal{I}_t^{(m)})_{t\ge 1}$ as a different model (since different indices yield different potentials) then $Z_{1:t}^{(m)}$ is the Bayesian evidence in favour of model $m$. Let us note, however, that $Z_{1:t}^{(m)}$ is not a direct indicator of the performance of worker $m$ as an optimizer. The fact that $Z_{1:t}^{(m_1)} > Z_{1:t}^{(m_2)}$ does not necessarily imply that the estimate of $\theta^\star$ computed from worker $m_1$ is quantifiably better than the estimate computed from worker $m_2$.}}. Given the index of the best performing sampler, which is given by
\begin{align*}
m_t^\star = \argmax_{m\in\{1,\ldots,M\}} Z_{1:t}^{(m),N},
\end{align*}
we obtain a maximum-a-posteriori (MAP) estimator,
\begin{align}\label{eq:maxEstimator}
\theta_t^{\star,N} = \argmax_{i\in\{1,\ldots,N\}} \mathsf{p}_t^{(m_t^\star),N}(\theta^{(i,m_t^\star)}),
\end{align}
where $\mathsf{p}_t^{(m_t^\star),N}(\theta)$ is the kernel density estimator \citep{silverman1998density,wand1994kernel} described in Remark~\ref{rem:KDE}. Note that we do \textit{not} construct the entire density estimator and maximize it. Since this operation is performed locally on the particles from the best performing sampler, it involves $\mathcal{O}(N^2)$ operations, where $N$ is the number of particles on a single worker, which is much smaller than the total number $MN$. The full procedure is outlined in Algorithm~\ref{alg:global}.
\begin{rem}\label{rem:KDE} Let $\sk:\Theta\to(0,\infty)$ be a bounded pdf with zero mean and finite second order moment, i.e., we have $\int_\Theta \|\theta\|_2^2 \sk(\theta) \md\theta < \infty$. We can use the particle system $\{\theta_t^{(i,m)}\}_{i=1}^N$ and the pdf $\sk(\cdot)$ to construct the kernel density estimator (KDE) of $\pi_t^{(m)}(\theta)$ as
\begin{align}
\mathsf{p}_t^{(m),N}(\theta) &= \frac{1}{N}\sum_{i=1}^N \sk(\theta - \theta_t^{(i,m)})\nonumber \\
&= (k^\theta,\pi_t^{(m),N}),\label{eq:KDE}
\end{align}
where $\sk^\theta(\theta') = \sk(\theta - \theta')$. Note that $\mathsf{p}_t^{(m),N}(\theta)$ is not a standard KDE because the particles $\{\theta_t^{(i,m)}\}_{i=1}^N$ are not i.i.d. samples from $\pi_t^{(m)}(\theta)$. Eq.~\eqref{eq:KDE}, however, suggests that the estimator, $\mathsf{p}_t^{(m),N}(\theta)$ converges when the approximate measure $\pi_t^{(m),N}$ does. See \citet{crisan2014particle} for an analysis of particle KDE's. \hfill $\blacksquare$	
\end{rem}

\section{Analysis}\label{sec:Analysis}

In this section, we provide some basic theoretical guarantees for Algorithm \ref{alg:global}. In particular, we prove results regarding a sampler on a single worker $m$. To ease the notation, we skip the superscript ${}^{(m)}$ in the rest of this section and simply note that results presented below hold for every $m\in\{1,\ldots,M\}$. All proofs are deferred to the Appendix.

When constrained to a single worker $m$, the approximation $\pi_t^{N}$ is provably convergent. In particular, we have the following results that hold for every worker $m = 1,\ldots,M$.

\begin{thm}\label{thm:LocalNodeConv}
{If the sequence $(G_t)_{t\geq 1}$ satisfies Assumption~\ref{assmp:integrability}} {and the jittering kernels satisfy Assumption~\ref{assmp:kernel}}, then, for any $\varphi\in B(\Theta)$, we have
\begin{align*}
\left\|\left(\varphi,\pi_t\right) - \left(\varphi,\pi_t^{N}\right)\right\|_p \leq \frac{c_{t,p} \|\varphi\|_\infty}{\sqrt{N}}
\end{align*}
for every $t = 1,\ldots,T$ and for any $p\geq 1$, where $c_{t,p} > 0$ is a constant independent of $N$.
\end{thm}
\begin{proof} See Appendix~\ref{proof:LocalNodeConv}.\end{proof}

Theorem~\ref{thm:LocalNodeConv} states that the samplers on local workers converge to their correct probability measures (for each $m$) with rate $\mathcal{O}(1/\sqrt{N})$, which is standard for Monte Carlo methods. Next we provide an upper bound for the random error $| (\varphi,\pi_t) - (\varphi,\pi_t^N) |$.

\begin{corollary}\label{corr:almostSureConv} 
Under the assumptions of Theorem~\ref{thm:LocalNodeConv}, for every $\varphi \in B(\Theta)$, we have
\begin{align*}
\left|(\varphi,\pi_t^{N}) - (\varphi,\pi_t)\right| \leq \frac{U_{t,\delta}}{N^{\frac{1}{2}-\delta}},
\quad \mbox{and} \quad 1 \le t \le T,
\end{align*}
where $U_{t,\delta}$ is an a.s. finite random variable and $0 < \delta < \frac{1}{2}$ is an arbitrary constant independent of $N$. In particular,
\begin{align}\label{eq:AlmostSureConv}
\lim_{N\to\infty} (\varphi,\pi_t^{N}) = (\varphi,\pi_t) \quad\quad \textnormal{a.s.}
\end{align}
for any $\varphi\in B(\Theta)$.
\end{corollary}
\begin{proof}
See Appendix~\ref{proof:corr:almostSureConv}.
\end{proof}
This result ensures that the \textit{random error} made by the estimators vanishes as $N\to\infty$. Moreover, it provides us with a rate $\mathcal{O}(1/\sqrt{N})$ since the constant $\delta > 0$ can be chosen arbitrarily small.

These results are important because they enable us to analyze the properties of the kernel density estimators constructed using the samples at each worker. In order to be able to do so, we need to impose regularity conditions on the sequence of densities $\pi_t(\theta)$ and the kernels we use to approximate them.
\begin{assmp}\label{assmp:densityLipschitz} 
For every $\theta \in \Theta$, the derivatives $\sD^\alpha \pi_t(\theta)$ exist and they are Lipschitz continuous, i.e., there is a constant $L_{\alpha,t} > 0$ such that
\begin{align*}
| \sD^\alpha \pi_t(\theta) - \sD^\alpha \pi_t(\theta')| \leq L_{\alpha,t} \|\theta - \theta'\|
\end{align*}
for all $\theta,\theta' \in \Theta$, $t = 1,\ldots,T$ and for all $\alpha = (\alpha_1,\ldots,\alpha_d) \in \{0,1\}^d$.
\end{assmp}

Note that for $\alpha = (0,\ldots,0)$ it is not hard to relate Assumption~\ref{assmp:densityLipschitz} directly to the cost function as we do in the following proposition.
\begin{prop}\label{prop:Lipschitz} 
Assume that we define the incremental cost functions
\begin{align*}
F_t(\theta) = \sum_{i \in \bigcup_{k=1}^t \cI_k} f_i(\theta)
\end{align*}
and there exists some $\ell_t$ such that
\begin{align*}
|F_t(\theta) - F_t(\theta')| \leq \ell_t \|\theta - \theta'\|,
\end{align*}
i.e., $F_t$ is Lipschitz. Assume there exists $F^\star_t = \min_{\theta\in\Theta} F_t(\theta)$ such that $|F^\star_t| < \infty$ and recall that $\pi_t(\theta) \propto \exp(-F_t(\theta))$. Then we have the following inequality,
\begin{align*}
|\pi_t(\theta) - \pi_t(\theta')| \leq \frac{\ell_t \exp(-F_t^\star)}{Z_{\pi_t}} \|\theta - \theta'\|_2
\end{align*}
where $Z_{\pi_t} = \int_\Theta \exp(-F_t(\theta)) \md \theta$.
\end{prop}
\begin{proof} See Appendix~\ref{proof:Lipschitz}. \end{proof}

Next, we state assumptions on the kernel $\sk$. We first note that the kernels in practice are defined with a bandwidth parameter $h\in\bR_+$. In particular, given a kernel $\sk$, we can define scaled kernels $\sk_h$ as
\begin{align*}
\sk_h(\theta) = h^{-d} \sk(h^{-1} \theta), \quad\quad h > 0,
\end{align*}
where, we recall, $d$ is the dimension of the parameter vector $\theta$. Hence, given $\sk$ we define a family of kernels $\{\sk_h, h\in\bR_+\}$.

\begin{assmp}\label{assmp:kernelFull} 
The kernel $\sk:\Theta \to (0,\infty)$ is a zero-mean bounded pdf, i.e., $\sk(\theta) \geq 0$ $\forall \theta\in\Theta$ and $\int \sk(\theta) \md \theta = 1$. The second moment of this density is bounded, i.e., $\int_\Theta \|\theta\|^2 \sk(\theta) \md \theta < \infty$. Finally, $\sD^\alpha \sk \in \sC_b(\Theta)$, i.e., $\|\sD^\alpha\sk\|_\infty~<~\infty$ for any $\alpha \in \{0,1\}^d$.
\end{assmp}
\begin{rem} We note that Assumption~\ref{assmp:kernelFull} implies that $\sD^\alpha \sk_h \in \sC_b(\Theta)$ and we have $\|\sD^\alpha \sk_h\|_\infty = \frac{1}{h^{d + |\alpha|}} \|\sD^\alpha \sk\|_\infty$ for any $h > 0$ and $\alpha\in\{0,1\}^d$.\hfill $\blacksquare$	
\end{rem}
We denote \textit{the kernel density estimator} defined using a scaled kernel $\sk_h$ and the empirical measure $\pi_t^N$ as $\sp_t^{h,N}(\theta)$. In particular, given a normalized kernel (a pdf) $\sk:\Theta\to(0,\infty)$, satisfying the assumptions in Assumption~\ref{assmp:kernelFull}, we can construct the KDE
\begin{align*}
\mathsf{p}_t^{h,N}(\theta) &= (\sk^{\theta}_h,\pi_t^{N}).
\end{align*}
where $\sk_h^\theta(\theta') = \sk_h(\theta - \theta')$ (see Remark~\ref{rem:KDE}). Now, we are ready to state the main results regarding the kernel density estimators, adapted from \citet{crisan2014particle}.

\begin{thm}\label{thm:uniformResult}
Choose
\begin{align}
 h = \left\lfloor{N^{\frac{1}{2 (d + 1)}}}\right\rfloor^{-1}
\end{align}
and denote $\sp_t^N(\theta) = \sp_t^{h,N}(\theta)$ (since $h = h(N)$). If Assumptions~\ref{assmp:integrability},~\ref{assmp:kernel},~\ref{assmp:densityLipschitz} and \ref{assmp:kernelFull} hold, and $\Theta$ is compact, then
\begin{align}\label{eq:thm:uniformResultRate}
\sup_{\theta\in\Theta} | \sp_t^N(\theta) - \pi_t(\theta)| \leq \frac{V_\varepsilon}{N^{\frac{1-\varepsilon}{2 (d + 1)}}}
\end{align}
where $V_\varepsilon \geq 0$ is an a.s. finite random variable and $0< \varepsilon < 1$ is a constant, both of which are independent of $N$ and $\theta$. In particular,
\begin{align}\label{eq:thm:uniformResult}
\lim_{N\to\infty} \sup_{\theta\in\Theta} | \sp_t^N(\theta) - \pi_t(\theta)| = 0 \quad \quad \textnormal{a.s.}
\end{align}
\end{thm}
\begin{proof}
It follows from the proof of Theorem~4.2 and Corollary~4.1 in \citet{crisan2014particle}. See Appendix~\ref{proof:uniformResult} for an outline.
\end{proof}
This theorem is a uniform convergence result, i.e., it holds uniformly in a compact parameter space $\Theta$. We note that Theorem~\ref{thm:uniformResult} specifies the dependence of the bandwidth $h$ on the number of Monte Carlo samples $N$ for convergence to be attained at that rate. Based on this result, we can relate the empirical maxima to the true maxima.
\begin{thm}\label{thm:FinalConvergenceRate} 
Let $\theta_t^{\star,N} \in \argmax_{i\in\{1,\ldots,N\}} \sp_t^{N}(\theta_t^{(i)})$ be an estimate of a global maximum of $\pi_t$ and let $\theta_t^{\star} \in \argmax_{\theta\in\Theta} \pi_t(\theta)$ be an actual global maximum. If $\Theta$ is compact, $\pi_t$ is continuous at $\theta_t^\star$ and Assumptions~\ref{assmp:integrability},~\ref{assmp:kernel},~\ref{assmp:densityLipschitz} and \ref{assmp:kernelFull} hold, then for $N$ sufficiently large
\begin{align*}
\pi_t(\theta_t^\star) - \pi_t(\theta_t^{\star,N}) \leq \frac{W_{t,d,\varepsilon}}{N^{\frac{1-\varepsilon}{2 (d + 1)}}}, \quad 1 \le t \le T,
\end{align*}
where $\varepsilon \in (0,1)$ is an arbitrarily small constant and $W_{t,d,\varepsilon}$ is an a.s. finite random variable, both independent of $N$.
\end{thm}
\begin{proof}
See Appendix~\ref{proof:thm:FinalConvergenceRate}.
\end{proof}

\begin{rem}
By choosing $t=T$, Theorem \ref{thm:FinalConvergenceRate} provides a convergence rate for the MAP estimator $\theta_T^\star$, which is also the approximate solution of problem \eqref{eqCostFunction}. \hfill \qedsymbol
\end{rem}


Theorem~\ref{thm:FinalConvergenceRate} also yields a convergence rate for the error $f(\theta_T^{\star,N}) - f(\theta^\star)$, where $f(\cdot)$ is the original cost function in problem \eqref{eqCostFunction}, provided that the prior is chosen so that $\pi_T(\theta) \propto \exp(-f(\theta))$ (see Remark~\ref{rem:Prior}).
\begin{corollary}\label{cor:ConvergenceRateCostFunc} 
Choose any 
$$
\theta^\star \in \argmin_{\theta\in\Theta} f(\theta) 
\quad \mbox{and} \quad
\theta_T^{\star,N} \in \argmax_{i\in\{1,\ldots,N\}} \sp_T^{N}(\theta_T^{(i)}).
$$
Under the same assumptions as in Theorem \ref{thm:FinalConvergenceRate}, if $\| f \|_\infty<\infty$ then we have
\begin{align*}
0 \le f(\theta_T^{\star,N}) - f(\theta^\star) \leq  \frac{\tilde{W}_{T,d,\varepsilon}}{N^{\frac{1-\varepsilon}{2 (d + 1)}}},
\end{align*}
where $\tilde{W}_{T,d,\varepsilon}$ is an a.s. finite random variable.
\end{corollary}
\begin{proof}
See Appendix~\ref{proof:ConvergenceRateCostFunc}.
\end{proof}

Finally, we obtain a convergence rate for the expected error.
\begin{corollary}\label{corr:ErgodicConvergenceRateCostFunc} 
Choose any 
$$
\theta^\star \in \argmin_{\theta\in\Theta} f(\theta) 
\quad \mbox{and} \quad
\theta_T^{\star,N} \in \argmax_{i\in\{1,\ldots,N\}} \sp_T^{N}(\theta_T^{(i)}).
$$
Under the same assumptions as in Theorem \ref{thm:FinalConvergenceRate}, if $\| f \|_\infty<\infty$ then we have
\begin{align*}
0 \le \bE[f(\theta_T^{\star,N})] - f(\theta^\star) \leq  \frac{C_{T,d,\varepsilon}}{N^{\frac{1-\varepsilon}{2 (d + 1)}}},
\end{align*}
where $C_{T,d,\varepsilon} = \bE[\tilde{W}_{T,d,\varepsilon}] < \infty$ is a constant independent of $N$.
\end{corollary}
\begin{proof}
The proof follows from Corollary~\ref{cor:ConvergenceRateCostFunc}, since $\tilde{W}_{T,d,\varepsilon}$ is an a.s. finite random variable.
\end{proof}

\subsection{Discussion}


{Theorem \ref{thm:FinalConvergenceRate} and Corollaries \ref{cor:ConvergenceRateCostFunc} and \ref{corr:ErgodicConvergenceRateCostFunc} go beyond standard results on the convergence of SMC methods. While the latter refer to the approximation of integrals (in the vein of Theorem \ref{thm:LocalNodeConv} and Corollary \ref{corr:almostSureConv}), Corollaries \ref{cor:ConvergenceRateCostFunc} and \ref{corr:ErgodicConvergenceRateCostFunc} directly address the convergence of the sequence of optimizers $\theta_T^{\star,N}$ and state that the proposed algorithm yields, with probability 1, an asymptotically optimal solution to problem \eqref{eqCostFunction} even if $f(\theta)$ is non-convex and presents multiple local and/or global minima. These results also provide explicit convergence rates that depend on the computational cost (the number of particles $N$) and the dimension $d$ of the search space.}

Note that the analyses available in the literature for most Monte Carlo optimization algorithms are purely asymptotical (see \citet{appel2004accelerated,ikonen2005application,miguez2010analysis,hu2012survey,zhou2013particle}, i.e., they do not provide explicit convergence rates. Moreover, they often rely on restrictive assumptions. For example, \citet{hu2012survey} and \citet{zhou2013particle} require that the objective function present a unique global minimum. More detailed analyses are carried out by \citet{zhou2013sequential} and \citet{miguez2013convergence}. However, the former falls short of providing explicit error rates for the sequence of optimizers (bounds are given for the total variation distance between the Boltzmann distributions and their SMC approximations in a SA scheme) and the latter relies on a sequential decomposition of the cost function that is not satisfied by $f(\theta)$ in problem \eqref{eqCostFunction}. Moreover, all the analytical results in these papers \citep{appel2004accelerated,ikonen2005application,miguez2010analysis,hu2012survey,zhou2013particle,zhou2013sequential,miguez2013convergence} are obtained for deterministic optimization problems where the objective function can be evaluated exactly, while Theorem \ref{thm:FinalConvergenceRate} and Corollaries \ref{cor:ConvergenceRateCostFunc} and \ref{corr:ErgodicConvergenceRateCostFunc}  hold for a more general stochastic optimization framework where $f(\theta)$ can only be estimated using mini-batches of data.

%
\section{Numerical Results}\label{sec:Exp}

In this section, we show numerical results for {three} optimization problems which are hard to solve with conventional methods. In the first example, we focus on minimizing a function with multiple global minima. The aim of this experiment is to show that, when the cost function has several global minima, the PSMCO algorithm can successfully populate with Monte Carlo samples the regions of $\Theta$ that contain these minima. In the second example, we tackle the minimization of a challenging cost function, with broad flat regions, for which standard stochastic gradient optimizers struggle. The third example involves a non-convex, non-smooth cost function and we use it to compare the proposed PSMCO scheme with a similar SMC-based optimization method proposed in \citet{stinis2012stochastic}.

%
\subsection{Minimization of a function with multiple global minima}\label{sec:Exp1}
\begin{figure*}
\begin{center}
\includegraphics[scale=0.41]{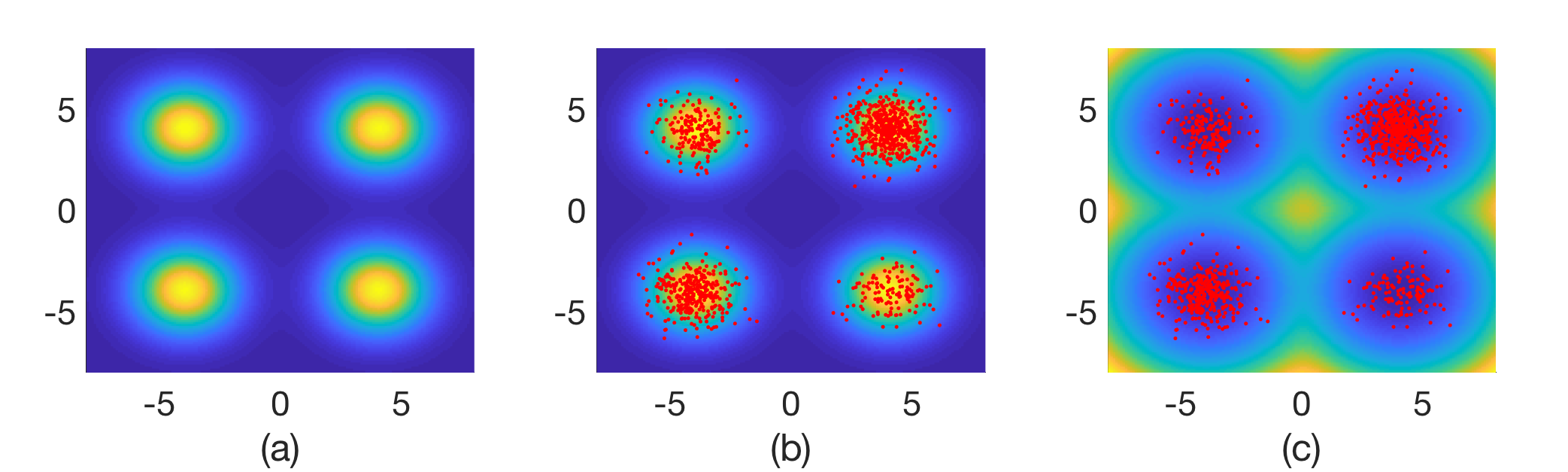}
\caption{An illustration of the performance of the proposed algorithm for a cost function with four global minima. (a) The plot of $\pi_T(\theta) \propto \exp(-f(\theta))$. The blue regions indicate low values. It can be seen that there are four global maxima. (b) Samples drawn by the PSMCO at a single time instant. (c) The plot of the samples together with the actual cost function $f(\theta)$.}
\label{fig:multGlobal}
\end{center}
\end{figure*}
In this experiment, we tackle the problem
\begin{align*}
\min_{\theta\in\bR^2} f(\theta), \textnormal{ where } \quad f (\theta) = \sum_{i=1}^n f_i(\theta)
\end{align*}
and
\begin{align*}
f_i(\theta) = -\frac{1}{\lambda} \log\left(\sum_{k=1}^4 \NPDF(\theta;m_{i,k}, R)\right),
\end{align*}
with $\lambda = 10$ and $R = r I_2$, with $I_2$ denoting the $2 \times 2$ identity matrix and $r = 0.2$. We choose the means $m_{i,k}$ randomly, namely $m_{i,k} \sim \NPDF(m_{i,k}; m_k, \sigma^2)$ where,
\begin{align*}
&m_1 = [4,4]^\top, \quad m_2 = [-4,-4]^\top,\\ &m_3 = [-4,4]^\top, \quad m_4 = [4,-4]^\top,
\end{align*}
and $\sigma^2 = 0.5$. This selection results in a cost function with four global minima. Such functions arise in many machine learning problems, see, e.g., \citet{mei2018landscape}. In this experiment, we have chosen $n = 1,000$. Although a small number for stochastic optimization problems, we note that each $f_i(\theta)$ represents a mini-batch in this scenario and we set $K= 1$ in the PSMCO algorithm. 

In order to run the algorithm, we choose a uniform prior measure $\pi_0(\theta) = \mathcal{U}([-a,a]\times [-a,a])$ with $a = 50$. It follows from Proposition~\ref{prop:RN} that the pdf that matches the cost function $f(\theta)$ can be written as
\begin{align*}
\pi_T(\theta) \propto \exp(-f(\theta)),
\end{align*}
and it has four global maxima. This pdf is displayed in Fig.~\ref{fig:multGlobal}(a). We run $M = 100$ samplers, with $N = 50$ particles each, yielding a total number of  $MN = 5,000$ particles. We choose a Gaussian jittering scheme; specifically, the jittering kernel is defined as
\begin{align}\label{eq:GaussianJitteringKernel}
\kappa(\mbox{d}\theta | \theta') = (1 - \epsilon_N) \delta_{\theta'}(\mbox{d}\theta) + \epsilon_N \NPDF(\theta; \theta',\sigma^2_j)\md\theta,
\end{align}
where $\epsilon_N = 1/\sqrt{N}$ and $\sigma_j^2 = 0.5$.

Some illustrative results can be seen from Fig.~\ref{fig:multGlobal}. To be specific, we have run independent samplers and plot all samples for this experiment (instead of estimating a minimum with the best performing sampler). From Fig.~\ref{fig:multGlobal}(b), it can be seen that the algorithm populates the regions surrounding all maxima with samples. Finally, Fig.~\ref{fig:multGlobal}(c) shows the location of the samples relative to the actual cost function $f(\theta)$. These plots illustrate how the algorithm ``locates'' multiple, distinct global maxima with independent samplers. Note different samplers can converge to different global maxima in practice --which is in agreement with the analysis provided in Section~\ref{sec:Analysis}.

%
\subsection{Minimization of the sigmoid function}

In this experiment, we address the problem,
\begin{align}\label{eq:SigmoidCost}
\min_{\theta\in\bR^2} f(\theta) := \sum_{i=1}^n (y_i - g_i(\theta))^2,
\end{align}
where
\begin{align*}
\quad g_i(\theta)= \frac{1}{1 + \exp(-\theta_1 - \theta_2 x_i)},
\end{align*}
with $x_i\in\bR$, $f_i(\theta) = (y_i - g_i(\theta))^2$ and $\theta = [\theta_1,\theta_2]^\top$. The function $g_i$ is called as the sigmoid function. Cost functions of the form in eq.~\eqref{eq:SigmoidCost} are widely used in nonlinear regression with neural networks in machine learning \citep{Bishop2006}.

In this experiment, we have $n = 100,000$. We choose $M = 25$ and $MN = 1,000$, leading to $N=40$ particles for every sampler. The mini-batch size is $K = 100$. The jittering kernel $\kappa$ is defined in the same way as in \eqref{eq:GaussianJitteringKernel}, where the Gaussian pdf has a variance chosen as the ratio of the dataset size $L$ to the mini-batch size $K$, i.e., $\sigma_j^2 = {n}/{K}$, which yields a rather large variance\footnote{Note that this is for efficient exploration of the global minima, which are hard to find for this example. A large jittering variance may not be adequate in practice when there are multiple minima close to each other, see, e.g., Section~\ref{sec:Exp1}.} $\sigma_j^2 = 1000$. To compute the maximum as described in Eq.~\eqref{eq:maxEstimator}, we use a Gaussian kernel density with bandwidth $h = \lfloor{N^{\frac{1}{6}}}\rfloor^{-1}$.

\begin{figure*}
\begin{center}
\includegraphics[scale=0.4]{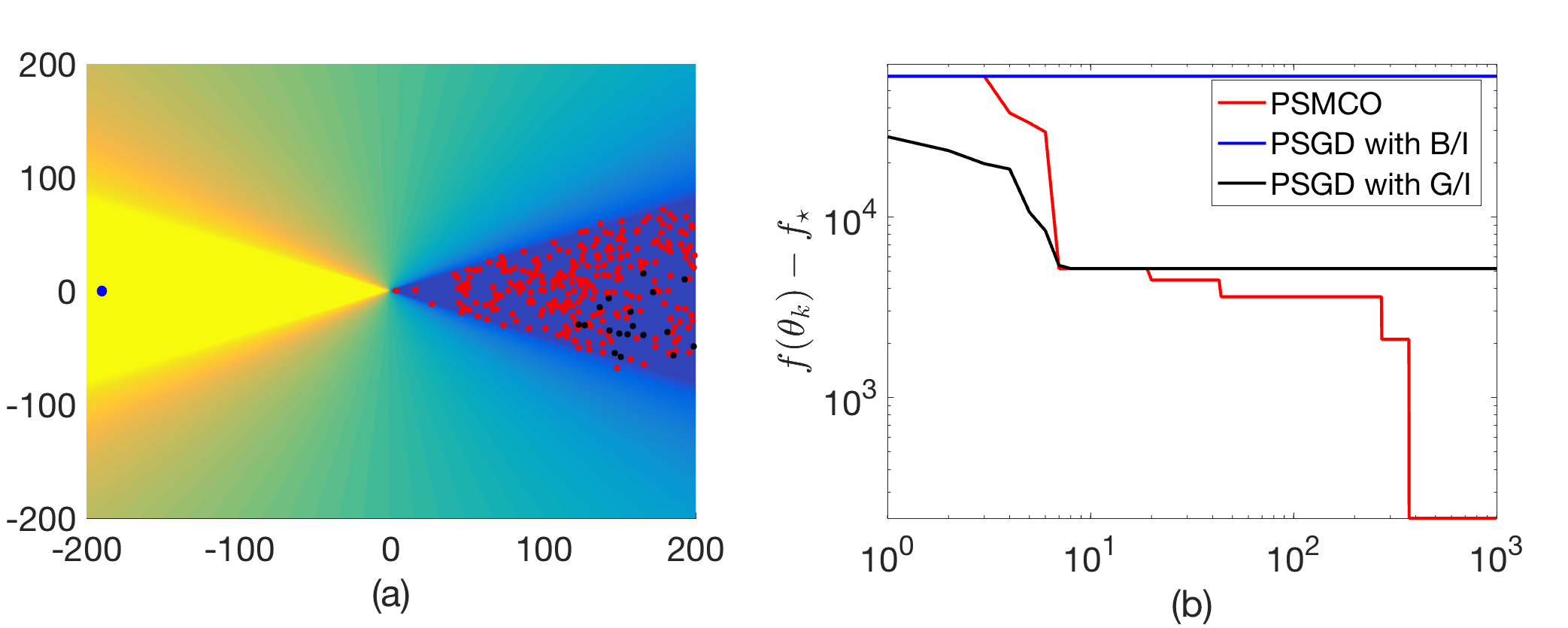}
\caption{(a) The cost function and a snapshot of samples from the 50th iteration of the PSMCO, PSGD with bad initialization (blue dot on the yellow area) and PSGD with good initialization (black dots on the blue area). (b) Performance of each algorithm: it can be seen that PSMCO first converges to the wide region with low values (blue region) and then jumps to the minimum. This is because the marginal likelihood estimate of the sampler close to the minimum dominates after a while. There is effectively full communication among samplers only to determine the minimizer.}
\label{fig:LogFunc}
\end{center}
\end{figure*}

In Fig.~\ref{fig:LogFunc} we compare the PSMCO algorithm with a parallel stochastic gradient descent (PSGD) scheme \citep{zinkevich2010parallelized} using $M$ optimizers. We note that, given a particular realization\footnote{For this experiment, we generate i.i.d. uniform realizations, $x_k \sim \mathcal{U}([-2.5,2.5])$ for $k=1, \ldots, n$.} of $(x_i)_{i=1}^n$, searching for a minimum of $f(\theta)$ may be a hard task. Fig.~\ref{fig:LogFunc}(a)  shows one such case, where the cost function has broad flat regions which make it difficult to find its maxima using gradient based methods unless their initialization is sufficiently good. Accordingly, we have run two instances of PSGD with ``bad'' and ``good'' initializations.

The bad initial point for PSGD can be seen from Fig~\ref{fig:LogFunc}(a), at $[-190,0]^\top$ (the blue dot). We initialize $M$ parallel SGD optimizers around $[-190,0]^\top$, each with a small zero-mean Gaussian perturbation with variance $10^{-8}$. This is a poor initialization because gradients are nearly zero in this region (yellow area in Fig.~\ref{fig:LogFunc}(a)). We refer to the PSGD algorithm starting from this point as PSGD with B/I, which refers to bad initialization. We also initialize the PSMCO from this region, with Gaussian perturbations around $[-190,0]^\top$, with the same small variance $\sigma_\textnormal{init}^2 = 10^{-8}$.

The ``good'' initialization for the PSGD is selected from a better region, namely around the point $[0,-100]^\top$, where gradient values actually contain useful information about the minimum. We refer to the PSGD algorithm starting from this point as PSGD with G/I.

The results can be seen in Fig.~\ref{fig:LogFunc}(b). We observe that the PSGD with good initialization (G/I) moves towards a better region, however, it gets stuck because the gradient becomes nearly zero. On the other hand, PSGD with B/I is unable to move at all, since it is initialized in a region where all gradients are negligible (which is true even for the mini-batch observations). The PSMCO method, on the other hand, searches the space effectively to find the global minimum, as depicted in Fig.~\ref{fig:LogFunc}(b).

{
\subsection{Constrained nonsmooth nonconvex optimization}
In this section, we compare the proposed PSMCO scheme to the method of \citet{stinis2012stochastic}, labeled here as `particle filtering for stochastic global optimization' (PFSGO), and the stochastic evolution strategies (SES) algorithm in \citet{salimans2017evolution} for a high-dimensional non-smooth and non-convex optimization problem. In particular, we apply this algorithms to numerically solve the problem
\begin{align}\label{eq:NonconvexFull}
\min_{\theta\in\Theta} \frac{1}{2} \|y - X^\top \theta\|^2 + \frac{\rho}{2} \sum_{i=1}^d P_{\lambda,\gamma}(\theta_i),
\end{align}
where $y \in \bR^n$, $X \in \bR^{d\times n}$, $\Theta = [-5,5]^d$, the dimension $d$ is set to different values (see below),  and $P_{\lambda,\gamma}:\bR\mapsto\bR$ is given by
\begin{align*}
P_{\lambda,\gamma}(x) = \left\{
\begin{aligned}
&\lambda|x|  && ~\text{if} ~|x|\leq \lambda,\\
&\tfrac{2\gamma\lambda|x|-x^2-\lambda^2}{2(\gamma-1)}  && ~\text{if}~ \lambda<|x|<\gamma \lambda,\\
&\tfrac{\lambda^2(\gamma+1)}{2}  && ~\text{if} ~|x|\geq\gamma\lambda,
\end{aligned}
\right.
\end{align*}
where $\lambda> 2$ and $\gamma > 0$. This problem formulation is useful for variable selection, see, e.g., \citet{fan2001variable} or \citet{lan2019accelerated}. It is easy to see that problem \eqref{eq:NonconvexFull} can be written as
\begin{align}\label{eq:NonconvexFull2}
\min_{\theta\in\Theta} \frac{1}{2} \sum_{i=1}^n (y_i - x_i^\top \theta)^2 + \frac{\rho}{2} \sum_{i=1}^d P_{\lambda,\gamma}(\theta_i),
\end{align}
where $y_i \in \bR$, and $x_i \in \bR^d$. This, in turn, makes the problem an instance of \eqref{eqCostFunction}, with $$f_i(\theta) = \frac{1}{2} (y_i - x_i^\top \theta)^2 + \frac{\tilde{\rho}}{2} \sum_{i=1}^d P_{\lambda,\gamma}(\theta_i),$$
and $\tilde{\rho} = \rho/n$.
} 

In this problem, we also test the single-worker version of the proposed optimization scheme. We refer to this algorithm simply as SMCO and it is obtained as the particular case of PSMCO with $M=1$. We use the usual jittering kernel of the form \eqref{eq:kernelDefn}
\begin{align*}
\kappa(\mbox{d}\theta | \theta') = (1 - \epsilon_N) \delta_{\theta'}(\mbox{d}\theta) + \epsilon_N \tau(\mbox{d}\theta | \theta'),
\end{align*}
where $\tau$ is a Gaussian kernel with covariance $C = \sigma^2 I_d$ for both methods. We also use the same Gaussian transition kernel for the PFSGO. Let us remark, though, that (unlike SMCO and PSMCO) the PFSGO scheme modifies all particles at every iteration, i.e., it uses $\tau(\cdot | \theta')$ instead of $\kappa(\cdot|\theta')$ for sampling. The SES scheme also uses $\tau$ in order to estimate the gradients.

We choose $\sigma^2 = 10^{-2}$ and $N = 100$. The mini-batch size is taken as $K = 1$ and the number of components is $n = 1,000$. For the PSMCO, we chose $M = 5$, so it essentially runs $5$ samplers with $20$ particles each while the SMCO scheme runs a single sampler with $N=100$ particles. For the regularization parameters, we choose $\tilde{\rho} = 1, \lambda = 10^{-3}$, and $\gamma = 2.01$. For the SES, we choose a small step size of $\alpha = 10^{-7}$ as larger values cause it to diverge. We simulate the data using a sparse parameter $\theta^\star$, where only three values are nonzero. We simulate the entries of the matrix $X$ as i.i.d. variates from $\NPDF(0,1)$ and compute $y = X^\top \theta^\star$. In order to compute the error for an iterate $\theta_k$ produced by any method, we compute
\begin{align*}
\mathsf{NMSE}(k) = \frac{\|\theta_k - \theta^\star\|^2}{\|\theta^\star\|^2}.
\end{align*}

The results can be seen in Fig.~\ref{fig:NonSmoothOpt}. We also plot the 0.5$\sigma$ curves around the error curves which are averaged over $1,000$ Monte Carlo runs. It can be seen that, for this particular example, the SMCO performs the best, while the PSMCO still outperforms the PFSGO. The SES basically is very slow due to the inefficiency of the gradient estimators for this problem.

\begin{figure}
\begin{center}
\includegraphics[width=0.95\linewidth]{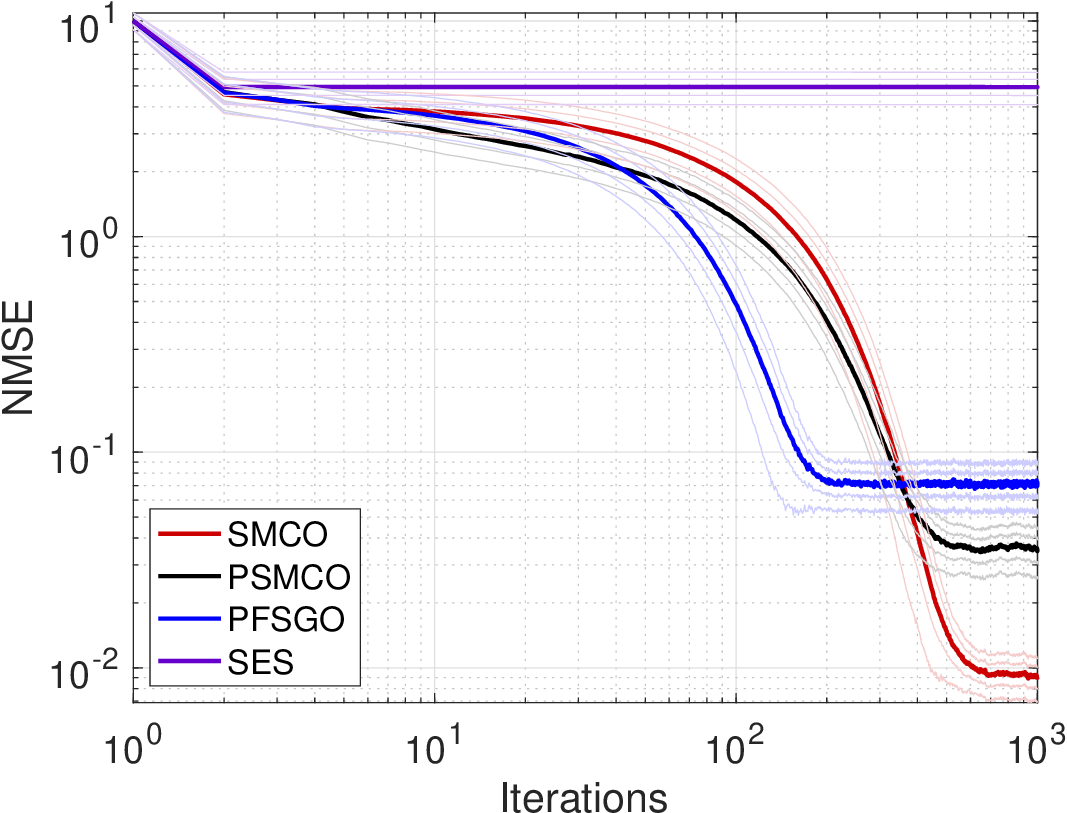}
\caption{Comparison of algorithms for problem \eqref{eq:NonconvexFull2} with $d = 10$, $N = 100$, and $n = 1,000$. It can be seen that the SMCO is the most efficient method for this problem and the PSMCO ($M = 5$) is the second best. Although PFSGO converges faster, the steady error that it attains is higher. The results are averaged over $1,000$ Monte Carlo runs.}
\label{fig:NonSmoothOpt}
\end{center}
\end{figure}

\begin{figure}
\begin{center}
\includegraphics[width=0.95\linewidth]{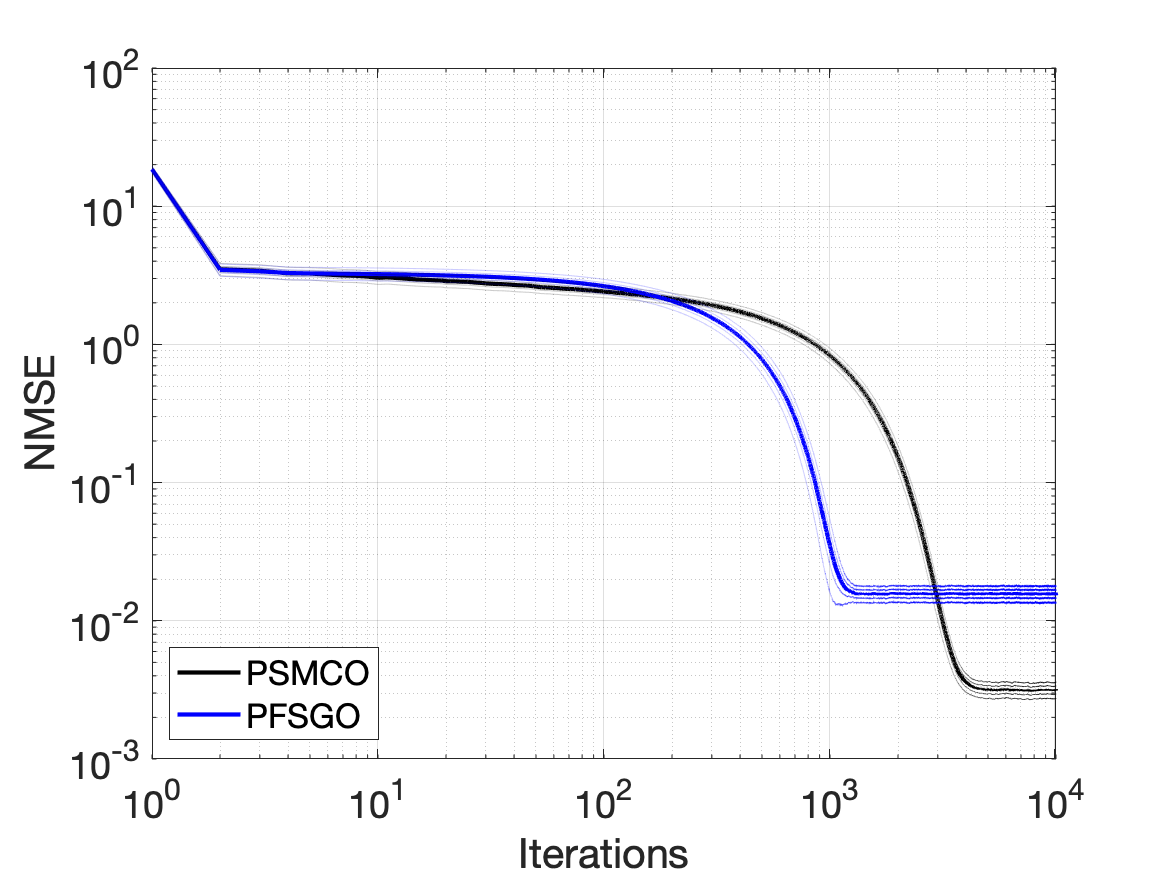}
\caption{Comparison of algorithms for problem \eqref{eq:NonconvexFull2}, with $d = 30$, $N = 1,000$, and $n = 10,000$ --only for the PSMCO ($M = 25$) and PFSGO schemes. The results are averaged over $1,000$ Monte Carlo runs.}
\label{fig:NonSmoothOpt2}
\end{center}
\end{figure}

To gain further insight, we also compare PSMCO ($M = 25$) and the PFSGO on a problem that is higher-dimensional, namely $d = 30$, and with more data points, $n = 10,000$. We set $\sigma^2 = 10^{-3}$ and leaving other parameters same as in the example with $d  = 10$.

Figure \ref{fig:NonSmoothOpt2} displays the results for this example. It can be seen that again the PSMCO algorithm converges to a point which has lower NMSE than the PFSGO. We believe that this is mainly due to the difference in the transition kernels. The PFSGO uses a full transition kernel where every particle is modified whereas jittering enables us to induce slower and more careful changes and also gives us a chance to keep a particle unmodified if it is in a good location.

\section{Conclusions}\label{sec:Conc}
We have proposed a parallel sequential Monte Carlo optimization algorithm which does not require the computation (either exact or approximate) of gradients and, therefore, can be applied to the minimization of challenging cost functions, e.g., with multiple global minima or with broad ``flat'' regions. The proposed method uses jittering kernels to propagate samples \citep{crisan2018nested} and particle kernel density estimators to find the minima \citep{crisan2014particle}, within a stochastic optimization setup. We have provided a detailed analysis of the proposed scheme. In particular, we have proved that it yields asymptotically optimal solutions to the stochastic optimization problem \eqref{eqCostFunction} (as the number of samples $N$ is increased) and we have computed explicit convergence rates for the resulting optimizers that depend on $N$ and the dimension of the search space, $d$. These results are new and improve on classical asymptotic analyses for Monte Carlo optimization methods, which typically lack convergence rates.

From a practical perspective, we argue that the parallel setting where each sampler uses a different configuration of the same dataset can be useful to improve the practical behaviour of the algorithm. To illustrate this  point, we have studied the numerical performance of the PSMCO algorithm in scenarios where gradient-based methods struggle to converge. In this work, we have focused on challenging but relatively low-dimensional cost functions. We leave the potential applications of our scheme to high-dimensional optimization problems as a future work. Also the design of an interacting extension of our method similar to particle islands \citep{verge2015parallel} can be potentially useful in more challenging settings.

%




%
\appendix
\section{Appendix}

%
\subsection{Proof of Proposition~\ref{prop:RN}}\label{proof:RN}

We prove this result by induction. For $t = 1$, let
\begin{align*}
\pi_1(\md\theta) = \pi_0(\md\theta) \frac{G_1(\theta)}{\int_\Theta G_1(\theta) \pi_0(\md\theta)} = \pi_0(\md \theta) \frac{G_1(\theta)}{(G_1,\pi_0)}.
\end{align*}
Since $G_1 \in B(\Theta)$ it follows that
\begin{align*}
\sup_{\theta\in\Theta} \left| \frac{G_1(\theta)}{(G_1,\pi_0)}\right| = \frac{\sup_{\theta\in\Theta} G_1(\theta)}{(G_1,\pi_0)} < \infty
\end{align*}
because of Assumption~\ref{assmp:integrability}. Hence $\pi_1 \ll \pi_0$ is a proper measure. Assume next, as an induction hypothesis, that $\pi_{T-1} \ll \pi_0$. Then
\begin{align*}
\pi_T(\md\theta) = \pi_{T-1}(\md\theta) \frac{G_T(\theta)}{(G_T,\pi_{T-1})}
\end{align*}
and Assumption~\ref{assmp:integrability} implies (again) that
\begin{align*}
\frac{\sup_{\theta\in\Theta} G_T(\theta)}{(G_T,\pi_{T-1})} < \infty,
\end{align*}
hence $\pi_T$ is proper and $\pi_T \ll \pi_0$. Therefore, the Radon-Nikodym derivative of the final measure $\pi_T$ w.r.t. the prior $\pi_0$ is
\begin{align*}
\frac{\md\pi_T}{\md\pi_0}(\theta) \propto \prod_{t=1}^T G_t(\theta) = \exp\left(-\sum_{i=1}^n f_i(\theta)\right).
\end{align*}
From here, it readily follows that maximizing this Radon-Nikodym derivative is equivalent to solving problem \eqref{eqCostFunction}.
\hfill$\blacksquare$

%
\subsection{Proof of Theorem~\ref{thm:LocalNodeConv}}\label{proof:LocalNodeConv}

We proceed by an induction argument. At time $t = 0$, the bound
\begin{align*}
\|(\varphi,\pi_0^N) - (\varphi,\pi_0) \|_p \leq \frac{c_{0,p} \|\varphi\|_\infty}{\sqrt{N}}
\end{align*}
is a straightforward consequence of the Marcinkiewicz–Zygmund inequality \citep{shiryaev1996} because the particles $\{\theta_0^{(i)}\}_{i=1}^N$ are i.i.d samples from $\pi_0$.

Assume now that, after iteration $t-1$, we have a particle set $\{{\theta}_{t-1}^{(i)}\}_{i=1}^N$ and the empirical measure $\pi^N_{t-1}(\mbox{d}\theta_{t-1}) = \frac{1}{N} \sum_{i=1}^N \delta_{\theta_{t-1}^{(i)}} (\mbox{d}\theta_{t-1})$, which satisfies
\begin{align}\label{eq:InductionHypothesis}
\left\|(\varphi,\pi_{t-1}) - (\varphi,\pi_{t-1}^N)\right\|_p \leq \frac{c_{t-1,p} \|\varphi\|_\infty}{\sqrt{N}}.
\end{align}

We first analyze the error in the jittering step. To this end, we construct the \textit{jittered} random measure
\begin{align*}
\hat{\pi}^N_t(\md \theta) = \frac{1}{N} \sum_{i=1}^N \delta_{\hat{\theta}_t^{(i)}}(\md \theta)
\end{align*}
and iterate the triangle inequality to obtain
\begin{align}
\|(\varphi,\pi_{t-1}) - (\varphi,\hat{\pi}_t^N)\|_p \leq&  \|(\varphi,\pi_{t-1}) - (\varphi,\pi^N_{t-1})\|_p \nonumber \\&+  \|(\varphi,\pi_{t-1}^N) - (\varphi,\kappa {\pi}_{t-1}^N)\|_p \nonumber \\ &+ \|(\varphi,\kappa \pi_{t-1}^N) - (\varphi,\hat{\pi}_t^N)\|_p,\label{eq:Decomp1}
\end{align}
where
\begin{align*}
\kappa \pi_{t-1}^N = \int \kappa(\mbox{d}\theta|\theta_{t-1}) \pi_{t-1}^N(\mbox{d}\theta_{t-1}) = \frac{1}{N} \sum_{i=1}^N \kappa(\mbox{d}\theta|{\theta_{t-1}^{(i)}}).
\end{align*}
The first term on the right hand side (rhs) of \eqref{eq:Decomp1} is bounded by the induction hypothesis \eqref{eq:InductionHypothesis}. For the second term, we note that,
\begin{align}
\left|(\varphi,\pi_{t-1}^N) - (\varphi,\kappa {\pi}_{t-1}^N)\right| &= \left| \frac{1}{N} \sum_{i=1}^N \varphi(\theta_{t-1}^{(i)}) - \frac{1}{N} \sum_{i=1}^N \int \varphi(\theta) \kappa(\mbox{d}\theta|{\theta_{t-1}^{(i)}}) \right|\nonumber \\
&= \left| \frac{1}{N} \sum_{i=1}^N \int \left(\varphi(\theta_{t-1}^{(i)}) - \varphi(\theta)\right) \kappa(\mbox{d}\theta|{\theta_{t-1}^{(i)}}) \right|\nonumber \\
&\leq \frac{1}{N} \sum_{i=1}^N \int \left| \varphi(\theta_{t-1}^{(i)}) - \varphi(\theta) \right|  \kappa(\mbox{d}\theta|{\theta_{t-1}^{(i)}})\nonumber \\
&\leq \frac{c_\kappa \|\varphi\|_\infty}{\sqrt{N}},\label{eq:ProofSecondTermBound}
\end{align}
where the last inequality follows from Assumption~\ref{assmp:kernel}. The upper bound in \eqref{eq:ProofSecondTermBound} is deterministic, so the inequality readily implies that
\begin{align}\label{eq:bound2}
\|(\varphi,\pi_{t-1}^N) - (\varphi,\kappa\pi_{t-1}^N)\|_p \leq \frac{c_\kappa \|\varphi\|_\infty}{\sqrt{N}}.
\end{align}

For the last term on the right-hand side of \eqref{eq:Decomp1}, we let $\cF_{t-1}$ be the $\sigma$-algebra generated by the random sequence $\{\theta_{0:t-1}^{(i)},\hat{\theta}_{1:t-1}^{(i)}\}_{i=1}^N$. Let us first note that
\begin{align*}
\bE\left[(\varphi,\hat{\pi}_t) | \cF_{t-1}\right] &= \frac{1}{N} \sum_{i=1}^N \bE\left[\varphi(\hat{\theta}_t^{(i)}) | \cF_{t-1} \right]\\
&= \frac{1}{N}\sum_{i=1}^N \int \varphi(\theta) \kappa(\md\theta|\theta_{t-1}^{(i)}) = (\varphi,\kappa\pi_{t-1}^N).
\end{align*}
Therefore, the difference $(\varphi,\hat{\pi}_t^N) -(\varphi,\kappa \pi_{t-1}^N)$ takes the form
\begin{align*}
(\varphi,\hat{\pi}_t^N) - (\varphi,\kappa\pi_{t-1}^N) = \frac{1}{N} \sum_{i=1}^N S^{(i)},
\end{align*}
where $S^{(i)} = \varphi(\hat{\theta}_t^{(i)}) - \bE[\varphi(\hat{\theta}_t^{(i)})|\cF_{t-1}]$, $i = 1,\ldots,N$, are zero-mean and conditionally independent random variables, with $|S^{(i)}| \leq 2 \|\varphi\|_\infty$. Then we readily obtain the bound
\begin{align}
\bE\left[\left.\left|(\varphi,\hat{\pi}_t^N) - (\varphi,\kappa\pi_{t-1}^N)\right|^p \right| \cF_{t-1}\right] &= \frac{1}{N^p} \bE\left[ \left.\left|\sum_{i=1}^N S^{(i)}\right|^p\right|\cF_{t-1}\right] \nonumber \\
&\leq \frac{B_{t,p} N^{\frac{p}{2}}\|\varphi\|_\infty^p}{N^p}.\label{eq:MZineq}
\end{align}
where the relation \eqref{eq:MZineq} follows from the Marcinkiewicz–Zygmund inequality
\citep{shiryaev1996} and $B_{t,p} < \infty$ is some constant independent of $N$. Taking unconditional expectations on both sides of \eqref{eq:MZineq} and then computing $(\cdot)^\frac{1}{p}$ yields
\begin{align}
\|(\varphi,\hat{\pi}_t^N) - (\varphi,\kappa\pi_{t-1}^N)\|_p \leq \frac{\hat{c}_{t,p} \|\varphi\|_\infty}{\sqrt{N}}.\label{eq:JitterFinalTermBd}
\end{align}
where $\hat{c}_{t,p} = B_{t,p}^{\frac{1}{p}}$ is a finite constant independent of $N$.
Therefore, taking together \eqref{eq:InductionHypothesis}, \eqref{eq:bound2} and \eqref{eq:JitterFinalTermBd} we have established that
\begin{align}\label{proof:JitteringBound}
\|(\varphi,\pi_{t-1}) - (\varphi,\hat{\pi}_t^N)\|_p \leq& \frac{c_{1,t,p} \|\varphi\|_\infty}{\sqrt{N}},
\end{align}
where $c_{1,t,p} = c_{t-1,p} + c_\kappa + \hat{c}_{t,p} < \infty$ is a finite constant independent of $N$.

Next, we have to bound the error after the weighting step. We recall that
\begin{align*}
\pi_t(\mbox{d}\theta) = \pi_{t-1}(\mbox{d}\theta) \frac{G_t(\theta)}{(G_t,\pi_{t-1})}
\end{align*}
and define
\begin{align*}
\tilde{\pi}_t^N(\mbox{d}\theta) = \hat{\pi}^N_{t}(\mbox{d}\theta) \frac{G_t(\theta)}{(G_t,\hat{\pi}^N_{t})}
\end{align*}
where $\tilde{\pi}_t^N$ denotes the \textit{weighted} measure. We first note that
\begin{align}
&|(\varphi,\pi_t) - (\varphi,\tilde{\pi}_t^N)| = \left| \frac{(\varphi G_t, \pi_{t-1})}{(G_t,\pi_{t-1})} - \frac{(\varphi G_t, \hat{\pi}_t^N)}{(G_t,\hat{\pi}_t^N)} \pm \frac{(\varphi G_t, \hat{\pi}^N_t)}{(G_t,\pi_{t-1})} \right|\nonumber \\
&\leq \frac{\left|(\varphi G_t, \pi_{t-1}) - (\varphi G_t, \hat{\pi}_t^N)\right| + \|\varphi\|_\infty |(G_t,\hat{\pi}_t^N) - (G_t,\pi_{t-1})|}{(G_t,\pi_{t-1})}.\label{eq:finalIneq}
\end{align}
Using Minkowski's inequality together with \eqref{proof:JitteringBound} and \eqref{eq:finalIneq} yields
\begin{align*}
\|(\varphi,\pi_t) - (\varphi,\tilde{\pi}_t^N)\|_p &\leq \frac{c_{1,t,p} \|\varphi G_t\|_\infty + c_{1,t,p} \|\varphi\|_\infty \|G_t\|_\infty}{(G_t,\pi_{t-1})\sqrt{N}}, \\
&\leq \frac{2 c_{1,t,p} \|\varphi\|_\infty \|G_t\|_\infty}{(G_t,\pi_{t-1})\sqrt{N}}
\end{align*}
where the second inequality follows from $\|\varphi G_t \|_\infty \leq \|\varphi\|_\infty \|G_t\|_\infty$. More concisely, we have
\begin{align}\label{proof:WeightingBound}
\|(\varphi,\pi_t) - (\varphi,\tilde{\pi}_t^N)\|_p \leq \frac{c_{2,t,p} \|\varphi\|_\infty}{\sqrt{N}}
\end{align}
where the constant
\begin{align*}
c_{2,t,p} = \frac{2 c_{1,t,p} \|G_t\|_\infty}{(G_t,\pi_{t-1})} < \infty
\end{align*}
is independent of $N$. Note that the assumptions on $(G_t)_{t\geq 1}$ imply that $(G_t,\pi_{t-1}) > 0$.

Finally, we bound the resampling step. Note that the resampling step consists of drawing $N$ i.i.d samples from $\tilde{\pi}_t^N$, i.e. $\theta_t^{(i)} \sim \tilde{\pi}_t^N$ i.i.d for $i = 1, \ldots,N$, and then constructing
\begin{align*}
\pi^N_t(\mbox{d}\theta) = \frac{1}{N} \sum_{i=1}^N \delta_{\theta_t^{(i)}}(\mbox{d}\theta).
\end{align*}
Since samples are i.i.d, as in the base case, we have,
\begin{align}\label{proof:ResamplingBound}
\|(\varphi,\tilde{\pi}_t^N) - (\varphi,\pi_t^N)\|_p \leq \frac{\tilde{c}_p \|\varphi\|_\infty}{\sqrt{N}},
\end{align}
for some constant $\tilde{c}_p < \infty$ independent of $N$. Now combining \eqref{proof:WeightingBound} and \eqref{proof:ResamplingBound}, we have the desired result,
\begin{align*}
\|(\varphi,\pi_t) - (\varphi,\pi_t^N)\|_p \leq \frac{c_t \|\varphi\|_\infty}{\sqrt{N}}
\end{align*}
where $c_t = c_{2,t,p} + \tilde{c}_p$ is a finite constant independent of $N$.
\hfill $\blacksquare$

%
\subsection{Proof of Corollary~\ref{corr:almostSureConv}}\label{proof:corr:almostSureConv} 

From Theorem~\ref{thm:LocalNodeConv}, we obtain
\begin{align*}
\|(\varphi,\pi_t) - (\varphi,\pi_t^N)\|_p \leq \frac{c_t \|\varphi\|_\infty}{\sqrt{N}},
\end{align*}
where $c_t < \infty$ is a constant independent of $N$. Let us choose $p\geq 4$ and $0 < \epsilon < 1$. We construct the nonnegative random variable
\begin{align*}
U_{t,\epsilon}^{p} = \sum_{N=1}^\infty N^{\frac{p}{2} - 1 - \epsilon} |(\varphi,\pi_t) - (\varphi,\pi_t^N)|^p.
\end{align*}
and use Fatou's lemma to obtain
\begin{align}
\bE[U_{t,\epsilon}^{p}] &\leq \sum_{N=1}^\infty N^{\frac{p}{2} - 1 - \epsilon} \bE\left[\left|(\varphi,\pi_t) - (\varphi,\pi^N_t)\right|^p \right],\nonumber \\
&\leq c^p \|\varphi\|_\infty^p \sum_{N=1}^\infty N^{- 1 - \epsilon} < \infty, \label{eq:FatouBound}
\end{align}
where the second inequality follows from Theorem~\ref{thm:LocalNodeConv}. The relationship \eqref{eq:FatouBound} implies that the r.v. $U^p_{t,\epsilon}$ is a.s. finite.

Finally, since (trivially) $N^{\frac{p}{2} - 1 - \epsilon} |(\varphi,\pi_t) - (\varphi,\pi_t^N)|^p \leq U_{t,\epsilon}^{p}$, we have
\begin{align}\label{eq:randomErrorRate}
|(\varphi,\pi_t) - (\varphi,\pi_t^N)| \leq \frac{U_{t,\delta}}{N^{\frac{1}{2} - \delta}},
\end{align}
where $\delta = \frac{1 + \epsilon}{p}$ and $U_{t,\delta} = (U_{t,\epsilon}^{p})^{\frac{1}{p}}$. Since $p\geq 4$ and $0 < \epsilon < 1$, it follows that $0 < \delta < \frac{1}{2}$. The almost sure convergence follows from \eqref{eq:randomErrorRate}. Taking $N\to\infty$ yields
\begin{align*}
\lim_{N\to\infty} |(\varphi,\pi_t) - (\varphi,\pi_t^N)| = 0 \quad \quad \textnormal{a.s.}
\end{align*}
\hfill $\blacksquare$

%
\subsection{Proof of Proposition~\ref{prop:Lipschitz}}\label{proof:Lipschitz}
Recall the assumption
\begin{align*}
|F_t(\theta) - F_t(\theta')| \leq \ell_t \|\theta -\theta'\|.
\end{align*}
We write $F_t^\star = \min_{\theta\in\Theta} F_t(\theta)$, which is assumed to be finite, but not necessarily nonnegative. We first prove that $\exp(-F_t(\theta))$ is also Lipschitz continuous. Note that we trivially have $\exp(-F_t(\theta)) \leq \exp(-F^\star_t)$ for all $\theta$ since $F_t(\theta) \geq F_t^\star$ for all $\theta$. Now consider any $(\theta,\theta') \in \Theta\times\Theta$. We first consider the case where $F_t(\theta) \leq F_t(\theta')$. We obtain
\begin{align}
0 < e^{-F_t(\theta)} - e^{-F_t(\theta')} &= e^{-F_t(\theta)} \left( 1 - e^{F_t(\theta) - F_t(\theta')}\right),\nonumber \\
&\leq e^{-F_t(\theta)} \left(1 - (1 + F_t(\theta) - F_t(\theta'))\right), \label{proof:prop2_1}
\end{align}
where we have used the inequality $e^a \geq 1 + a$. Therefore, we readily obtain from \eqref{proof:prop2_1}
\begin{align}
0< e^{-F_t(\theta)} - e^{-F_t(\theta')} &\leq e^{-F_t(\theta)} \left(F_t(\theta') - F_t(\theta)\right),\nonumber \\
&\leq e^{-F_t^\star} \left(F_t(\theta') - F_t(\theta)\right) \\
&= e^{-F_t^\star} |F_t(\theta') - F_t(\theta)|,\label{proof:prop2_2}
\end{align}
since $F_t(\theta) \leq F_t(\theta')$. Next, assume otherwise, i.e., $F_t(\theta) \geq F_t(\theta')$. In this case, we can also show using the same line of reasoning that
\begin{align}\label{proof:prop2_3}
e^{-F_t(\theta')} - e^{-F_t(\theta)} &\leq e^{-F_t^\star} \left(F_t(\theta) - F_t(\theta')\right) \\
&= e^{-F_t^\star} |F_t(\theta') - F_t(\theta)|,
\end{align}
since $F_t(\theta) \geq F_t(\theta')$. Therefore, we can conclude (combining \eqref{proof:prop2_2} and \eqref{proof:prop2_3}) that
\begin{align*}
|e^{-F_t(\theta)} - e^{-F_t(\theta')}| \leq e^{-F_t^\star} |F_t(\theta') - F_t(\theta)| \leq e^{-F_t^\star} \ell_t \|\theta - \theta'\|,
\end{align*}
where the last inequality holds because $F_t$ is Lipschitz. Finally recall that
\begin{align*}
\pi_t(\theta) = \frac{e^{-F_t(\theta)}}{Z_{\pi_t}},
\end{align*}
where we denote $Z_{\pi_t} = \int_\Theta e^{-F_t(\theta)}\md\theta$. We straightforwardly obtain
\begin{align*}
|\pi_t(\theta) - \pi_t(\theta')| \leq \frac{1}{Z_{\pi_t}} e^{-F_t^\star} \ell_t \|\theta - \theta'\|.
\end{align*}
\hfill $\blacksquare$

%
\subsection{Proof of Theorem~\ref{thm:uniformResult}}\label{proof:uniformResult}

Using the proof of Theorem~4.2 and Corollary~4.1 in \citet{crisan2014particle}, we obtain
\begin{align*}
\sup_{\theta\in\Theta} | \sp_t^N(\theta) - \pi_t(\theta)| \leq \frac{V_{1,\varepsilon}}{\left\lfloor{N^{\frac{1}{2 (d + 1)}}}\right\rfloor^{1-\varepsilon}},
\end{align*}
where $V_{1,\varepsilon}$ is an a.s. finite random variable. Noting that
\begin{align*}
\sup_{a \geq 1} \frac{a}{\left\lfloor a \right\rfloor} = 2,
\end{align*}
we obtain
\begin{align*}
\sup_{\theta\in\Theta} | \sp_t^N(\theta) - \pi_t(\theta)| \leq \frac{V_\varepsilon}{N^{\frac{1-\varepsilon}{2 (d + 1)}}},
\end{align*}
where $V_\varepsilon = 2 V_{1,\varepsilon}$ is an almost surely finite random variable. $\blacksquare$

%
\subsection{Proof of Theorem~\ref{thm:FinalConvergenceRate}}\label{proof:thm:FinalConvergenceRate}

Recall that $\pi_t(\theta)$ is a probability density w.r.t. the Lebesgue measure. Choose $\theta_t^\star \in \arg\max_{\theta\in\Theta} \pi_t(\theta)$ and construct the ball 
$$
B_{t,n}^\star := B\left( \theta_t^\star, \frac{1}{n} \right) \subset \Theta
$$
where $n \ge 1$ is a positive integer. We assume, without loss of generality, that $B_{t,1}^\star \subseteq \Theta$ and denote
$$
\pi_t(B_{t,n}^\star) = \int_{B_{t,n}^\star} \pi_t(\theta) \md \theta 
\quad \mbox{and} \quad
\pi_t^N(B_{t,n}^\star) = \int_{B_{t,n}^\star} \pi_t^N(\md \theta).
$$
Also recall that the grid of points generated by the SMC sampler at time $t$ is $\{ \theta_t^{(i)} \}_{1 \le i \le N} \subset \Theta$ and the estimate of $\theta_t^\star$ obtained from the grid is denoted
\begin{equation}
\theta_t^{\star,N} \in \arg\max_{\theta \in \{ \theta_t^{(i)} \}_{1 \le i \le N} } {\sf p}_t^N(\theta),
\label{eqA6--1}
\end{equation}
where ${\sf p}_t^N(\theta)$ is the kernel density estimator of $\pi_t$. Our argument to prove Theorem \ref{thm:FinalConvergenceRate} proceeds in two steps:
\begin{enumerate}
\item We show that, for any given $n \ge 1$, one can a.s. find $N$ sufficiently large to ensure that $\{ \theta_t^{(i)} \}_{1 \le i \le N} \cap B_{t,n}^\star \ne \emptyset$, i.e., that there are points of the grid contained in the ball $B_{t,n}^\star$. Moreover, we deduce an inequality that relates the radius $n^{-1}$ of the ball $B_{t,n}^\star$ with the number of necessary particles $N$.
  
\item From the existence of at least one particle $\theta_t^{(i)}$ inside $B_{t,n}^\star$ and the assumption that $\pi_t(\theta)$ is Lipschitz, we deduce bounds for the differences $|\pi_t(\theta_t^\star)-\pi_t(\theta_t^{(i)})|$ and $|\pi_t(\theta_t^{\star,N})-\pi_t(\theta_t^{(i)})|$, and, as a consequence, for the approximation error $|\pi_t(\theta_t^{\star,N})-\pi_t(\theta_t^\star)|$.
\end{enumerate}

\subsubsection{The ball $B_{t,n}^\star$ is a.s. non-empty}

Since $\pi_t(\theta)$ is assumed continuous at every $\theta_t^\star\in\arg\max_{\theta\in\Theta} \pi_t(\theta)$, we have $\pi_t(B_{t,n}^\star)>0$. Therefore, for every $n<\infty$, Theorem \ref{thm:uniformResult} ensures that there exists $N_n$ (a.s. finite) such that for all $N \ge N_n$,
\begin{equation}
\left|
	\pi_t^N(B_{t,n}^\star) - \pi_t(B_{t,n}^\star)
\right| < \frac{
	U_{t,\delta}
}{
	N^{\frac{1}{2}-\delta}
} < \frac{ \pi_t(B_{t,n}^\star)}{2},
\label{eqA6-0}
\end{equation}
where $U_{t,\delta}$ is an a.s. finite random variable and $\delta \in (0,\frac{1}{2})$ is an arbitrarily small constant (both independent of $N$). Moreover, the second inequality in \eqref{eqA6-0} implies that
\begin{equation}
\pi_t^N(B_{t,n}^\star) > \frac{\pi_t(B_{t,n}^\star)}{2} > 0.
\label{eqA6-1}
\end{equation}
Therefore, for all $N>N_n$ there exists at least one integer $i_b \in \{1, \ldots, N\}$ such that $\theta_t^{(i_b)} \in B_{t,n}^\star$.

To be specific, since $\pi_t(\theta)$ is a density w.r.t. the Lebesgue measure, we can readily find a lower bound for the integral $\pi_t(B_{t,n}^\star)$, namely
\begin{equation}
\frac{\pi_t(B_{t,n}^\star)}{2} > \frac{1}{2} \text{Leb}\left( B_{t,n}^\star \right) \times \inf_{\theta \in B_{t,n}^\star} \pi_t(\theta) > c_{t,d} n^{-d}
\nonumber 
\end{equation}
where $\text{Leb}(B_{t,n}^\star) = \frac{\pi^{\frac{d}{2}}}{\Gamma\left( \frac{d}{2}+1 \right) n^d}$ is the Lebesgue measure of the $d$-dimensional ball with radius $n^{-1}$, $\Gamma(\cdot)$ is Euler's gamma function and 
$$
c_{t,d} := \frac{\pi^{\frac{d}{2}}}{2 \Gamma\left( \frac{d}{2}+1 \right) n^d} \times \inf_{\theta \in B_{t,1}^\star} \pi_t(\theta) > 0.
$$ 
Therefore, for any given $n<\infty$, if we choose $N$ such that $0 < \frac{
	U_{t,\delta}
}{
	N^{\frac{1}{2}-\delta}
} < c_{t,d}n^{-d}
$, i.e.,
\begin{equation}
N \ge N_n := \left(
	\frac{
		U_{t,\delta}
	}{
		c_{t,d}
	}
\right)^{\frac{2}{1-2\delta}} n^{\frac{2d}{1-2\delta}}
\label{eqA6-2}
\end{equation}
then the inequalities \eqref{eqA6-0} and \eqref{eqA6-1} hold a.s. (note that $N_n<\infty$ a.s. because $n<\infty$ and $U_{t,\delta}<\infty$ a.s.).

\subsubsection{Error bounds}

Choose $i_b \in \{1, \ldots, N\}$ such that $\theta_t^{(i_b)} \in B_{t,n}^\star$. Such index exists a.s. whenever $N$ satisfies the inequality \eqref{eqA6-2}. Let us recall the construction of the estimate $\theta_t^{\star,N}$ from expression \eqref{eqA6--1} and denote
\begin{equation}
\hat \theta_t^{\star,N} \in \arg\max_{\theta\in\Theta} {\sf p}_t^N(\theta).
\nonumber
\end{equation}

Let $L_t<\infty$ be the Lipschitz constant of the pdf $\pi_t(\theta)$. Since $\theta_t^{(i_b)} \in B_{t,n}^\star$, we readily obtain the upper bound
\begin{equation}
\pi_t(\theta_t^\star) - \pi_t(\theta_t^{(i_b)}) < L_t n^{-1}
\nonumber
\end{equation}
and, therefore,
\begin{equation}
\pi_t(\theta_t^\star) - L_t n^{-1} < \pi_t(\theta_t^{(i_b)}). 
\label{eqA6-3}
\end{equation}
However, using Theorem \ref{thm:uniformResult} we obtain
\begin{equation}
\left|
	\pi_t(\theta_t^{(i_b)}) - {\sf p}_t^N(\theta_t^{(i_b)})
\right| < \frac{
	V_{t,\varepsilon}
}{
	N^{\frac{1-\varepsilon}{2(d+1)}}
},
\label{eqA6-4}
\end{equation}
where $\varepsilon\in(0,1)$ is an arbitrarily small constant and $V_{t,\varepsilon}$ is an a.s. finite random variable, both independent of $N$. Combining \eqref{eqA6-3} and \eqref{eqA6-4} yields
\begin{equation}
{\sf p}_t^N(\theta_t^{(i_b)}) > \pi_t(\theta_t^{(i_b)}) - \frac{
	V_{t,\varepsilon}
}{
	N^{\frac{1-\varepsilon}{2(d+1)}}
} > \pi_t(\theta_t^\star) - L_t n^{-1} - \frac{
	V_{t,\varepsilon}
}{
	N^{\frac{1-\varepsilon}{2(d+1)}}
}
\nonumber 
\end{equation}
and, as a consequence,
\begin{equation}
{\sf p}_t^N(\theta_t^{\star,N}) \ge {\sf p}_t^N(\theta_t^{(i_b)}) > \pi_t(\theta_t^\star) - L_t n^{-1} - \frac{
	V_{t,\varepsilon}
}{
	N^{\frac{1-\varepsilon}{2(d+1)}}
}.
\label{eqA6-5}
\end{equation}

Moreover, using Theorem \ref{thm:uniformResult} again, we find that
\begin{equation}
\left|
	\pi_t(\hat \theta_t^{\star,N}) - {\sf p}_t^N(\hat \theta_t^{\star,N})
\right| < \frac{
	V_{t,\varepsilon}
}{
	N^{\frac{1-\varepsilon}{2(d+1)}}
},
\label{eqA6-6}
\end{equation}
with the same constant $\varepsilon\in(0,1)$ and a.s. finite random variable $V_{t,\varepsilon}$ as in \eqref{eqA6-5}. Since $\pi_t(\hat \theta_t^{\star,N}) \le \pi_t(\theta_t^\star)$, the inequality \eqref{eqA6-6} implies that
\begin{equation}
{\sf p}_t^N(\hat \theta_t^{\star,N}) < \pi_t(\theta_t^\star) + \frac{
	V_{t,\varepsilon}
}{
	N^{\frac{1-\varepsilon}{2(d+1)}}
}
\nonumber 
\end{equation}
and, since ${\sf p}_t^N(\theta_t^{\star,N}) \le {\sf p}_t^N(\hat \theta_t^{\star,N})$, we arrive at
\begin{equation}
{\sf p}_t^N(\theta_t^{\star,N}) < \pi_t(\theta_t^\star) + \frac{
	V_{t,\varepsilon}
}{
	N^{\frac{1-\varepsilon}{2(d+1)}}
}.
\label{eqA6-7}
\end{equation}

Taking the inequalities \eqref{eqA6-5} and \eqref{eqA6-7} together, we readily obtain the uniform bound (for $\theta \in \Theta$)
\begin{equation}
\left|
	\pi_t(\theta_t^\star) - {\sf p}_t^N(\theta_t^{\star,N})
\right| < \frac{
	V_{t,\varepsilon}
}{
	N^{\frac{1-\varepsilon}{2(d+1)}}
} + L_t n^{-1}
\label{eqA6-8}
\end{equation}
and a simple triangle inequality then yields
\begin{eqnarray}
\left|
	\pi_t(\theta_t^{\star,N}) - \pi_t(\theta_t^\star) 
\right| &\le& \left|
	\pi_t(\theta_t^{\star,N}) - {\sf p}_t^N(\theta_t^{\star,N})
\right| + \left|
	{\sf p}_t^N(\theta_t^{\star,N}) - \pi_t(\theta_t^\star)
\right| \nonumber \\
&<& \frac{
	2V_{t,\varepsilon}
}{
	N^{\frac{1-\varepsilon}{2(d+1)}}
} + L_t n^{-1},
\label{eqA6-9}
\end{eqnarray}
where the second inequality follows from \eqref{eqA6-8} and yet another application of Theorem \ref{thm:uniformResult}. 

The inequality \eqref{eqA6-9} holds for any pair of integers $(N,n)$ that satisfies the relationship \eqref{eqA6-2}. For any given $N$, sufficiently large for 
$$
n_N := \sup\left\{ m \in \mathbb{N}: m^{-1} > \left( \frac{ U_{t,\delta} }{ c_{t,d} } \right)^{\frac{1}{d}} \frac{1}{N^{\frac{1 - 2\delta}{2d}}} \right\}
$$ 
to be well defined, the pair consisting of $N$ and $n=n_N$ satisfies  \eqref{eqA6-2}, while
\begin{equation}
n_N^{-1} \le 2 \left( \frac{ U_{t,\delta} }{ c_{t,d} } \right)^{\frac{1}{d}} \frac{1}{N^{\frac{1 - 2\delta}{2d}}}.
\label{eqA6-10}
\end{equation}
Hence, if we substitute $n=n_N$ in the inequality \eqref{eqA6-9} and then apply the inequality \eqref{eqA6-10} we arrive at
\begin{equation}
\left|
	\pi_t(\theta_t^{\star,N}) - \pi_t(\theta_t^\star) 
\right| < \frac{
	2V_{t,\varepsilon}
}{
	N^{\frac{1-\varepsilon}{2(d+1)}}
}  + 2 \left( 
	\frac{ 
		U_{t,\delta} 
	}{ 
		c_{t,d} 
	} 
\right)^{\frac{1}{d}} \frac{
	L_t
}{
	N^{\frac{1 - 2\delta}{2d}}
},
\label{eqA6-11}
\end{equation}
where $V_{t,\varepsilon}$ and $U_{t,\delta}$ are a.s. finite, and $L_t$ and $c_{t,d}$ are finite. The constants $\varepsilon \in (0,1)$ and $\delta\in (0,1/2)$ can be chosen arbitrarily small. Hence, if we let $0<\delta = \varepsilon / 2<\frac{1}{2}$, the r.h.s. of \eqref{eqA6-11} can be upper bounded, which results in the bound
\begin{equation}
 \left|
	\pi_t(\theta_t^{\star,N}) - \pi_t(\theta_t^\star) 
\right| < \frac{
	W_{t,d,\varepsilon}
}{
	N^{\frac{1-\varepsilon}{2(d+1)}}
},
\nonumber
\end{equation}
where
$$
W_{t,d,\varepsilon} = 2\left[
	V_{t,\varepsilon} + \left( 
		\frac{ 
			U_{t,\delta(\varepsilon)} 
		}{ 
			c_{t,d} 
		} 
	\right)^{\frac{1}{d}} L_t
\right] < \infty \quad \mbox{a.s.}
$$

%
\subsection{Proof of Corollary~\ref{cor:ConvergenceRateCostFunc}}\label{proof:ConvergenceRateCostFunc}

Recall that
\begin{align*}
\|f\|_\infty = \sup_{\theta\in\Theta} |f(\theta)| < \infty.
\end{align*}
Note that Theorem~\ref{thm:FinalConvergenceRate} implies that
\begin{align}
0 \le e^{-f(\theta^\star)} - e^{- f(\theta_T^{\star,N})} \leq \frac{W_{T,d,\varepsilon} Z_{\pi_T}}{N^\frac{1}{2(d+1)}},\label{proof:ConvergenceRateCostFuncAux1}
\end{align}
where $Z_{\pi_T}$ is the normalizing constant of $\pi_T$. Next, we lower bound the left-hand side of \eqref{proof:ConvergenceRateCostFuncAux1} as
\begin{align}
e^{-f(\theta^\star)} - e^{- f(\theta_T^{\star,N})} &= e^{- f(\theta_T^{\star,N})} \left(e^{f(\theta_T^{\star,N}) - f(\theta^\star)} - 1\right) \nonumber\\
&\geq e^{-\|f\|_\infty} (f(\theta_T^{\star,N}) - f(\theta^\star))\label{proof:ConvergenceRateCostFuncAux2}
\end{align}
where the last inequality follows from the relationships
$$
e^{-f(\theta_T^{\star,N})} \geq e^{-\|f\|_\infty}
$$
(since $f(\theta_T^{\star,N}) \leq \|f\|_\infty$) and $e^a \geq a + 1$ for $a \in \bR$. Combining \eqref{proof:ConvergenceRateCostFuncAux1} and \eqref{proof:ConvergenceRateCostFuncAux2}, we obtain
\begin{align*}
f(\theta_T^{\star,N}) - f(\theta^\star) \leq \frac{\tilde{W}_{T,d,\varepsilon}}{{N^\frac{1}{2(d+1)}}}
\end{align*}
where
\begin{align*}
\tilde{W}_{T,d,\varepsilon} = Z_{\pi_T} W_{T,d,\varepsilon} e^{\|f\|_\infty}
\end{align*}
is a.s. finite.

%

\bibliographystyle{spbasic}
\bibliography{draft}

\begin{thebibliography}{54}
\providecommand{\natexlab}[1]{#1}
\providecommand{\url}[1]{{#1}}
\providecommand{\urlprefix}{URL }
\expandafter\ifx\csname urlstyle\endcsname\relax
  \providecommand{\doi}[1]{DOI~\discretionary{}{}{}#1}\else
  \providecommand{\doi}{DOI~\discretionary{}{}{}\begingroup
  \urlstyle{rm}\Url}\fi
\providecommand{\eprint}[2][]{\url{#2}}

\bibitem[{Akyildiz and M{\'\i}guez(2020)}]{akyildiz2020nudging}
Akyildiz {\"O}D, M{\'\i}guez J (2020) Nudging the particle filter. Statistics
  and Computing 30(2):305--330

\bibitem[{Akyildiz et~al(2017)Akyildiz, Mari{\~n}o, and
  M{\'\i}guez}]{akyildiz2017adaptive}
Akyildiz OD, Mari{\~n}o IP, M{\'\i}guez J (2017) Adaptive noisy importance
  sampling for stochastic optimization. In: Computational Advances in
  Multi-Sensor Adaptive Processing (CAMSAP), 2017 IEEE 7th International
  Workshop on, IEEE, pp 1--5

\bibitem[{Akyildiz et~al(2018)Akyildiz, Elvira, and
  Miguez}]{akyildiz2018proximal}
Akyildiz OD, Elvira V, Miguez J (2018) {T}he {I}ncremental {P}roximal {M}ethod:
  {A} {P}robabilistic {P}erspective. In: Proc.~IEEE International Conference on
  Acoustics, Speech and Signal Processing (ICASSP), Calgary, Canada

\bibitem[{Alquier et~al(2016)Alquier, Friel, Everitt, and
  Boland}]{alquier2016noisy}
Alquier P, Friel N, Everitt R, Boland A (2016) Noisy {M}onte {C}arlo:
  {C}onvergence of {M}arkov chains with approximate transition kernels.
  Statistics and Computing 26(1-2):29--47

\bibitem[{Appel et~al(2004)Appel, Labarre, and
  Radulovic}]{appel2004accelerated}
Appel M, Labarre R, Radulovic D (2004) On accelerated random search. SIAM
  Journal on Optimization 14(3):708--731

\bibitem[{Bach and Perchet(2016)}]{bach2016highly}
Bach F, Perchet V (2016) Highly-smooth zero-th order online optimization. In:
  Conference on Learning Theory, pp 257--283

\bibitem[{Bertsekas(2011)}]{bertsekas2011incremental}
Bertsekas DP (2011) Incremental gradient, subgradient, and proximal methods for
  convex optimization: A survey. Optimization for Machine Learning 2010:1--38

\bibitem[{Bishop(2006)}]{Bishop2006}
Bishop CM (2006) Pattern Recognition and Machine Learning. Springer-Verlag New
  York, Inc., Secaucus, NJ, USA

\bibitem[{Bottou et~al(2018)Bottou, Curtis, and
  Nocedal}]{bottou2018optimization}
Bottou L, Curtis FE, Nocedal J (2018) Optimization methods for large-scale
  machine learning. Siam Review 60(2):223--311

\bibitem[{Chen et~al(2016)Chen, Carlson, Gan, Li, and Carin}]{chen2016bridging}
Chen C, Carlson D, Gan Z, Li C, Carin L (2016) Bridging the gap between
  stochastic gradient {MCMC} and stochastic optimization. In: Artificial
  Intelligence and Statistics, pp 1051--1060

\bibitem[{Chen and Wild(2015)}]{chen2015randomized}
Chen R, Wild S (2015) Randomized derivative-free optimization of noisy convex
  functions. arXiv preprint arXiv:150703332

\bibitem[{Conn et~al(2009)Conn, Scheinberg, and Vicente}]{conn2009introduction}
Conn AR, Scheinberg K, Vicente LN (2009) Introduction to derivative-free
  optimization, MPS-SIAM Series on Optimization, vol~8. SIAM

\bibitem[{Crisan and M{\'\i}guez(2014)}]{crisan2014particle}
Crisan D, M{\'\i}guez J (2014) Particle-kernel estimation of the filter density
  in state-space models. Bernoulli 20(4):1879--1929

\bibitem[{Crisan and Miguez(2018)}]{crisan2018nested}
Crisan D, Miguez J (2018) Nested particle filters for online parameter
  estimation in discrete-time state-space markov models. Bernoulli
  24(4A):3039--3086

\bibitem[{Del~Moral(2004)}]{del2004feynman}
Del~Moral P (2004) Feynman-Kac formulae: {G}enealogical and interacting
  particle systems with applications. Springer

\bibitem[{Del~Moral and Doisy(1999)}]{del1999maslov}
Del~Moral P, Doisy M (1999) Maslov idempotent probability calculus, {I}. Theory
  of Probability \& Its Applications 43(4):562--576

\bibitem[{Del~Moral et~al(2006)Del~Moral, Doucet, and
  Jasra}]{del2006sequential}
Del~Moral P, Doucet A, Jasra A (2006) Sequential {M}onte {C}arlo samplers.
  Journal of the Royal Statistical Society: Series B (Statistical Methodology)
  68(3):411--436

\bibitem[{Douc and Capp{\'e}(2005)}]{douc2005comparison}
Douc R, Capp{\'e} O (2005) Comparison of resampling schemes for particle
  filtering. In: Image and Signal Processing and Analysis, 2005. ISPA 2005.
  Proceedings of the 4th International Symposium on, IEEE, pp 64--69

\bibitem[{Duchi et~al(2011)Duchi, Hazan, and Singer}]{adaGrad}
Duchi J, Hazan E, Singer Y (2011) Adaptive subgradient methods for online
  learning and stochastic optimization. Journal of Machine Learning Research
  12(Jul):2121--2159

\bibitem[{Elvira et~al(2017)Elvira, M{\'\i}guez, and
  Djuri{\'c}}]{elvira2017adapting}
Elvira V, M{\'\i}guez J, Djuri{\'c} PM (2017) Adapting the number of particles
  in sequential monte carlo methods through an online scheme for convergence
  assessment. IEEE Transactions on Signal Processing 65(7):1781--1794

\bibitem[{Fan and Li(2001)}]{fan2001variable}
Fan J, Li R (2001) Variable selection via nonconcave penalized likelihood and
  its oracle properties. Journal of the American statistical Association
  96(456):1348--1360

\bibitem[{Ghadimi and Lan(2013)}]{ghadimi2013stochastic}
Ghadimi S, Lan G (2013) Stochastic first-and zeroth-order methods for nonconvex
  stochastic programming. SIAM Journal on Optimization 23(4):2341--2368

\bibitem[{G{\"u}rb{\"u}zbalaban et~al(2015)G{\"u}rb{\"u}zbalaban, Ozdaglar, and
  Parrilo}]{gurbuzbalaban2015random}
G{\"u}rb{\"u}zbalaban M, Ozdaglar A, Parrilo P (2015) Why random reshuffling
  beats stochastic gradient descent. arXiv preprint arXiv:151008560

\bibitem[{Hansen and Ostermeier(2001)}]{hansen2001completely}
Hansen N, Ostermeier A (2001) Completely derandomized self-adaptation in
  evolution strategies. Evolutionary computation 9(2):159--195

\bibitem[{Hu et~al(2012)Hu, Wang, Zhou, Fu, and Marcus}]{hu2012survey}
Hu J, Wang Y, Zhou E, Fu MC, Marcus SI (2012) A survey of some model-based
  methods for global optimization. In: Optimization, Control, and Applications
  of Stochastic Systems, Springer, pp 157--179

\bibitem[{Ikonen et~al(2005)Ikonen, Najim, and
  Del~Moral}]{ikonen2005application}
Ikonen E, Najim K, Del~Moral P (2005) Application of genealogical decision
  trees for open-loop tracking control. IFAC Proceedings Volumes 38(1):288--293

\bibitem[{Kingma and Ba(2014)}]{kingma2014adam}
Kingma DP, Ba J (2014) Adam: A method for stochastic optimization.
  arXiv:14126980

\bibitem[{Kirkpatrick et~al(1983)Kirkpatrick, Gelatt, and
  Vecchi}]{kirkpatrick1983optimization}
Kirkpatrick S, Gelatt CD, Vecchi MP (1983) Optimization by simulated annealing.
  {S}cience 220(4598):671--680

\bibitem[{Lan and Yang(2019)}]{lan2019accelerated}
Lan G, Yang Y (2019) Accelerated stochastic algorithms for nonconvex finite-sum
  and multiblock optimization. SIAM Journal on Optimization 29(4):2753--2784

\bibitem[{Liu et~al(2016)Liu, Cheng, and Shi}]{liu2016particle}
Liu B, Cheng S, Shi Y (2016) Particle filter optimization: A brief
  introduction. In: International Conference on Swarm Intelligence, Springer,
  pp 95--104

\bibitem[{Mari{\~n}o and M{\'\i}guez(2007)}]{marino2007monte}
Mari{\~n}o IP, M{\'\i}guez J (2007) {M}onte {C}arlo method for multiparameter
  estimation in coupled chaotic systems. Physical Review E 76(5):057203

\bibitem[{Mei et~al(2018)Mei, Bai, and Montanari}]{mei2018landscape}
Mei S, Bai Y, Montanari A (2018) The landscape of empirical risk for nonconvex
  losses. The Annals of Statistics 46(6A):2747--2774

\bibitem[{Homem-de Mello and Bayraksan(2014)}]{homem2014monte}
Homem-de Mello T, Bayraksan G (2014) Monte carlo sampling-based methods for
  stochastic optimization. Surveys in Operations Research and Management
  Science 19(1):56--85

\bibitem[{Miguez(2010)}]{miguez2010analysis}
Miguez J (2010) Analysis of a sequential {M}onte {C}arlo method for
  optimization in dynamical systems. Signal Processing 90(5):1609--1622

\bibitem[{M{\'\i}guez et~al(2013)M{\'\i}guez, Crisan, and
  Djuri{\'c}}]{miguez2013convergence}
M{\'\i}guez J, Crisan D, Djuri{\'c} PM (2013) On the convergence of two
  sequential monte carlo methods for maximum a posteriori sequence estimation
  and stochastic global optimization. Statistics and Computing 23(1):91--107

\bibitem[{Morse and Stanley(2016)}]{morse2016simple}
Morse G, Stanley KO (2016) Simple evolutionary optimization can rival
  stochastic gradient descent in neural networks. In: Proceedings of the
  Genetic and Evolutionary Computation Conference 2016, ACM, pp 477--484

\bibitem[{Nesterov and Spokoiny(2011)}]{nesterov2011random}
Nesterov Y, Spokoiny V (2011) Random gradient-free minimization of convex
  functions. Tech. rep., Universit{\'e} catholique de Louvain, Center for
  Operations Research and Econometrics (CORE)

\bibitem[{Pereyra et~al(2015)Pereyra, Schniter, Chouzenoux, Pesquet, Tourneret,
  Hero, and McLaughlin}]{pereyra2015survey}
Pereyra M, Schniter P, Chouzenoux E, Pesquet JC, Tourneret JY, Hero AO,
  McLaughlin S (2015) A survey of stochastic simulation and optimization
  methods in signal processing. IEEE Journal of Selected Topics in Signal
  Processing 10(2):224--241

\bibitem[{Robbins and Monro(1951)}]{RobbinsMonro}
Robbins H, Monro S (1951) A stochastic approximation method. Annals of
  Mathematical Statistics 22:400--407

\bibitem[{Robert and Casella(2004)}]{robert2004monte}
Robert CP, Casella G (2004) Monte {C}arlo statistical methods. John Wiley \&
  Sons

\bibitem[{Salimans et~al(2017)Salimans, Ho, Chen, Sidor, and
  Sutskever}]{salimans2017evolution}
Salimans T, Ho J, Chen X, Sidor S, Sutskever I (2017) Evolution strategies as a
  scalable alternative to reinforcement learning. arXiv preprint
  arXiv:170303864

\bibitem[{Shamir(2016)}]{shamir2016without}
Shamir O (2016) Without-replacement sampling for stochastic gradient methods.
  In: Advances in Neural Information Processing Systems, pp 46--54

\bibitem[{Shiryaev(1996)}]{shiryaev1996}
Shiryaev AN (1996) Probability. Springer

\bibitem[{Silverman(1998)}]{silverman1998density}
Silverman BW (1998) Density estimation for statistics and data analysis.
  Routledge

\bibitem[{Spall(2005)}]{spall2005introduction}
Spall JC (2005) Introduction to stochastic search and optimization: estimation,
  simulation, and control, vol~65. John Wiley \& Sons

\bibitem[{Stinis(2012)}]{stinis2012stochastic}
Stinis P (2012) Stochastic global optimization as a filtering problem. Journal
  of Computational Physics 231(4):2002--2014

\bibitem[{Verg{\'e} et~al(2015)Verg{\'e}, Dubarry, Del~Moral, and
  Moulines}]{verge2015parallel}
Verg{\'e} C, Dubarry C, Del~Moral P, Moulines E (2015) On parallel
  implementation of sequential monte carlo methods: the island particle model.
  Statistics and Computing 25(2):243--260

\bibitem[{Wand and Jones(1994)}]{wand1994kernel}
Wand MP, Jones MC (1994) Kernel smoothing. Chapman and Hall/CRC

\bibitem[{Welling and Teh(2011)}]{welling2011bayesian}
Welling M, Teh YW (2011) Bayesian learning via stochastic gradient langevin
  dynamics. In: Proceedings of the 28th international conference on machine
  learning (ICML-11), pp 681--688

\bibitem[{Wibisono et~al(2012)Wibisono, Wainwright, Jordan, and
  Duchi}]{wibisono2012finite}
Wibisono A, Wainwright MJ, Jordan MI, Duchi JC (2012) Finite sample convergence
  rates of zero-order stochastic optimization methods. In: Advances in Neural
  Information Processing Systems, pp 1439--1447

\bibitem[{Wierstra et~al(2014)Wierstra, Schaul, Glasmachers, Sun, Peters, and
  Schmidhuber}]{wierstra2014natural}
Wierstra D, Schaul T, Glasmachers T, Sun Y, Peters J, Schmidhuber J (2014)
  Natural evolution strategies. The Journal of Machine Learning Research
  15(1):949--980

\bibitem[{Zhou and Chen(2013)}]{zhou2013sequential}
Zhou E, Chen X (2013) Sequential monte carlo simulated annealing. Journal of
  Global Optimization 55(1):101--124

\bibitem[{Zhou et~al(2013)Zhou, Fu, and Marcus}]{zhou2013particle}
Zhou E, Fu MC, Marcus SI (2013) Particle filtering framework for a class of
  randomized optimization algorithms. IEEE Transactions on Automatic Control
  59(4):1025--1030

\bibitem[{Zinkevich et~al(2010)Zinkevich, Weimer, Li, and
  Smola}]{zinkevich2010parallelized}
Zinkevich M, Weimer M, Li L, Smola AJ (2010) Parallelized stochastic gradient
  descent. In: Advances in neural information processing systems, pp 2595--2603

\end{thebibliography}



\end{document}